\documentclass[10pt,onecolumn,journal]{IEEEtran}
\IEEEoverridecommandlockouts
\usepackage[pass,letterpaper]{geometry} %US letter format

\usepackage{cite}
\ifCLASSINFOpdf
\else
\fi
\usepackage{graphicx}
\usepackage{epstopdf}
\usepackage{color,soul}
\usepackage{amsfonts}
\usepackage{amsmath,amsthm,amssymb, cite}
\usepackage[style]{fncychap}
\usepackage{extarrows}
\usepackage{tikz}
\usepackage{multirow}
\usepackage{cases}
\usepackage{subfig}
\usetikzlibrary{arrows,shapes}
\usetikzlibrary{shadows}

\newtheorem{Definition}{Definition}

\newtheorem{Lemma}{Lemma}
\newtheorem{Proof}{Proof}
\newtheorem{Theorem}{Theorem}
\newtheorem{Remark}{Remark}
\newtheorem{Corollary}{Corollary}

\begin{document}
%\tableofcontents
%\Large{1. How can we achieve positive secrecy rates for both users in the two way wiretap channel?}
%\vspace{5mm}
%
\title{Achieving both positive secrecy rates of the users in two-way wiretap channel by individual secrecy}
\vspace{5mm}
%
%\Large{3. Exploring positive secrecy rates in two-way wiretap channel with individual secrecy}

\author{\IEEEauthorblockN{Chao Qi\IEEEauthorrefmark{1},
		Bin Dai\IEEEauthorrefmark{1},
		and Xiaohu Tang\IEEEauthorrefmark{1} %\IEEEmembership{Member, IEEE}
	}
	
	\IEEEauthorblockA{\IEEEauthorrefmark{1}Information Security and National Computing Grid Laboratory,
		Southwest Jiaotong University, Chengdu, China\\ E-mail: chaoqi@my.swjtu.edu.cn, daibin@home.swjtu.edu.cn, xhutang@home.swjtu.edu.cn}
}

\maketitle

\vspace{5mm}
\begin{abstract}
	In this paper, the individual secrecy of two-way wiretap channel is investigated, where two legitimate users' messages are separately guaranteed secure against an external eavesdropper. 
	For one thing, in some communication scenarios, the joint secrecy is impossible to achieve both positive secrecy rates of two users. For another, the individual secrecy satisfies the secrecy demand of many practical communication systems.
    Thus, firstly, an achievable secrecy rate region is derived for the general two-way wiretap channel with individual secrecy. In a deterministic channel, the region with individual secrecy is shown to be larger than that with joint secrecy.  
	Secondly, outer bounds on the secrecy capacity region are obtained for the general two-way wiretap channel and for two classes of special two-way wiretap channels. 
	The gap between inner and outer bounds on the secrecy capacity region is explored via the binary input two-way wiretap channels and the degraded Gaussian two-way wiretap. 
	Most notably, the secrecy capacity regions are established for the XOR channel and the degraded Gaussian two-way wiretap channel. 
	Furthermore, the secure sum-rate of the degraded Gaussian two-way wiretap channel under the individual secrecy constraint is demonstrated to be strictly larger than that under the joint secrecy constraint. 	
\end{abstract}

\IEEEpeerreviewmaketitle

\section{Introduction}\label{Sec: Introduction}
With the wide usage of the wireless networks nowadays, the security of wireless communication has become a crucial issue. 
Due to the open nature of the wireless channel, the wireless links are more vulnerable to eavesdropping. 
However, in the dynamic wireless network, the traditional cryptography faces many challenges in handling the security problem, such as complex key distribution and management.
By contrast, information theoretic secrecy guarantees secure communication against the eavesdropper even with unlimited computational power. 
In 1975, Wyner \cite{wyner1975wire} introduced information theoretic secrecy to a noisy degraded broadcast channel and demonstrated that secure communication is possible without any shared key beforehand. 
Thereafter, information theoretic secrecy, a more powerful approach to wireless secure transmission, has attracted intensive attention \cite{csiszar1978broadcast,Leung1978Gaussian,liang08multiple,lai2008relay,dai2015multiple}.

As one of the classic multi-user channels, two-way channel models a large range of bidirectional communications, where two users exchange messages with each other through a common channel. 
For instance, two users talk with each other simultaneously via a full-duplex telephone networks; the power control centre (e.g. electricity company) interchanges information with the user via a smart grid network.
The reliable communication of two-way channel was first studied by Shannon in \cite{shannon1961two}, where inner and outer bounds on channel capacity region were presented.  
Later, Tekin and Yener \cite{tekin2007gaussian} investigated the security along with reliability of the two-way channel in the presence of an external eavesdropper, which is referred to the two-way wiretap channel.
Mainly, the two-way wiretap channel is explored in two secrecy criteria. 
One is the {\em weak secrecy}, requiring that the rate of information leakage to the eavesdropper vanishes. 
For the two-way wiretap channel with weak secrecy, both inner and outer bounds on the secrecy capacity were obtained. 
For the inner bound on the secrecy capacity, Tekin and Yener  \cite{tekin2007gaussian,tekin2008general,tekin2010correction} and El Gamal et al. \cite{el2013achievable} respectively derived the achievable secrecy rate region for the Gaussian two-way wiretap channel and the general two-way wiretap channel. 
Specifically, reference \cite{el2013achievable} improves the results in \cite{tekin2007gaussian,tekin2008general,tekin2010correction} by a hybrid coding scheme combining the cooperative jamming and secret-key exchange mechanism. 
The outer bound on the secrecy capacity region of the degraded Gaussian two-way wiretap channel was studied in \cite{he2013role}. 
The other secrecy criterion is the {\em strong secrecy}, demanding that the information leakage to the eavesdropper, rather than the leakage rate, goes to zero. 
Regarding to the difficulty of studying strong secrecy, as we know, only Pierrot et al. \cite{pierrot2011strongly} provided an achievable secrecy rate region with the strong secrecy of the general two-way wiretap channel.

So far, all the previous works on no matter weak secrecy or strong secrecy focus on the {\em joint secrecy} of two-way wiretap channel, assuring security of two legitimate users' confidential messages together. 
However, if either of the legitimate users' outputs is a degraded version of the eavesdropper's output, achieving positive secrecy rates at both legitimate users is impossible with the joint secrecy (the details will be explained in the Lemma 1 in Section II). 
Such scenario is quite common, for instance the eavesdropper stays closer to the transmitter than the receiver does, as a result the legitimate receiver encounters more interferences and noises through the long distance transmission than the eavesdropper does. 
To achieve positive secrecy rates at both legitimate users, we introduce the {\em individual secrecy} of the two-way wiretap channel. Roughly speaking, {\em individual secrecy} requires that the rate of information leakage from each confidential message to the eavesdropper is made vanish. 
Comparatively, individual secrecy can be achieved by positive secrecy rates at both legitimate users. 
In fact, the individual secrecy constraint is also practical in other scenarios \cite{bhattad2005weakly,Kadhe2014Weakly}. 
For example, the secrecy criterion with the same definition is proposed in a multicast network \cite{bhattad2005weakly}, where a source node sends a set of message packets through the multicast network to the destination. 
The security in \cite{bhattad2005weakly} requires that wiretapper gains no information about each packet, while still potentially obtains no meaningful information about the source. 
Under this secrecy constraint, the multicast capacity can be achieved \cite{bhattad2005weakly}. 
Whereas, if the information leakage of all the packets goes to zero (the joint secrecy), it is impossible to achieve the multicast capacity \cite{Cai2002Secure}. Thus, the individual secrecy gains an advantage over the joint secrecy in \cite{bhattad2005weakly}.

Based on the analysis above, we investigate the individual secrecy of the two-way wiretap channel in this paper. 
Firstly, we derive an achievable secrecy rate region of the general two-way wiretap channel under the individual secrecy constraint.
In order to illustrate the intuition of the result, a deterministic channel is provided to show that the achievable secrecy rate region under the individual secrecy constraint is strictly larger than that with the joint secrecy in \cite{el2013achievable}.
Secondly, outer bounds on the secrecy capacity region are established for the general two-way wiretap channel and for two classes of special two-way wiretap channels.
Further, the gap between the inner and outer bounds on the secrecy capacity region is explored via two cases: the binary input two-way wiretap channels and the degraded Gaussian channel.
Most notably, we obtain the secrecy capacity region of the degraded Gaussian two-way wiretap channel under individual secrecy constraint. To the best of our knowledge, it is the first time to determine the secrecy capacity region for any kind of two-way wiretap channel in the literature.

The organization of this paper is as follows. In Section \ref{Sec: System model}, we introduce the two-way wiretap channel with the individual secrecy. 
In Section \ref{Sec: Individual}, we present our results of the general two-way wiretap channel with the individual secrecy. 
A deterministic two-way wiretap channel is also given to illustrate the intuition behind the results. 
In Section \ref{Sec: Example}, we investigate binary-input two-way wiretap channels and the degraded Gaussian two-way wiretap channel with the individual secrecy.
In the final section, we give the conclusions.

\section{System Model}\label{Sec: System model}
Before discussing the system model, note that in this paper, we use capital letters, lower case letters and calligraphic letters to denote the random variables, sample values and alphabets, respectively. 
A similar convention is applied to the random vectors and their sample values. For example, $X^n$ denotes a random $n$-vector $(X_1, X_2, \cdots, X_n)$, and $x^n$ is a sample vector value in $X^n$. 

\begin{figure}[!htbp]
	\centering
	\includegraphics[width=0.65\textwidth]{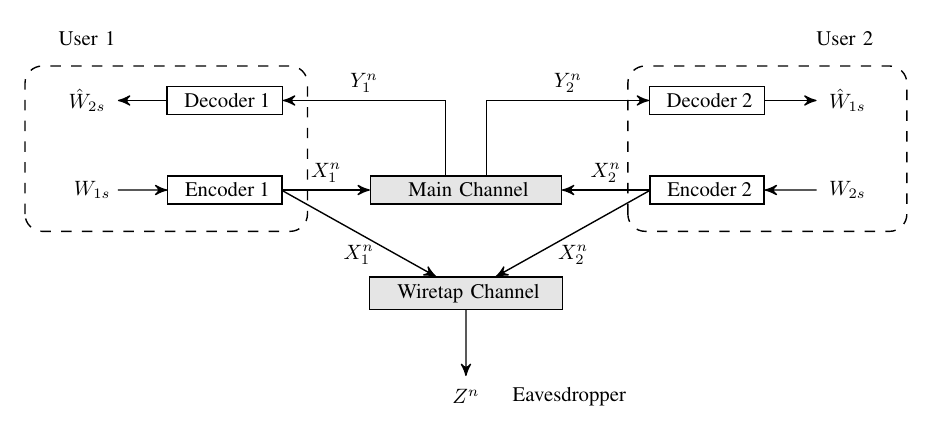}
	\caption{Two-way channel with an external eavesdropper.}
	\label{Fig: TWWT_channel_Joint}
\end{figure}    

In this paper, we study the two-way wiretap channel as shown in Fig. \ref{Fig: TWWT_channel_Joint}, where two legitimate users intend to exchange confidential messages with each other in the presence of an external eavesdropper. Particularly, we focus on the full-duplex scenario, where each of the legitimate users can send and receive messages simultaneously on the same degree of freedom.
%Note that, the two-way wiretap channel combines two the broadcast wiretap channels: one is from user 1 to user 2 and the eavesdropper; the other is from user 2 to user 1 and the eavesdropper. 

Suppose $\mathcal{W}_{1s}$, $\mathcal{W}_{2s}$ are two message sets; $\mathcal{X}_1$, $\mathcal{X}_2$ are the finite channel input alphabets at user 1 and user 2; $\mathcal{Y}_1$, $\mathcal{Y}_2$, $\mathcal{Z}$ are the channel output alphabets at user 1, user 2 and the eavesdropper, respectively. 
The discrete memoryless two-way wiretap channel is characterized by the transition probability distribution $p(y_1,y_2,z|x_1,x_2)$, where $x_1\in \mathcal{X}_1$, $x_2\in \mathcal{X}_2$ are the channel inputs from user 1 and 2; $y_1\in \mathcal{Y}_1$, $y_2\in \mathcal{Y}_2$ and $z\in \mathcal{Z}$ are channel outputs at user 1, user 2 and the eavesdropper. 
More specifically, the legitimate user $i$ wants to transmit a confidential message $W_i\in \mathcal{W}_{i}$ to the other user. The corresponding codeword $X^n_i\in \mathcal{X}_i$ is sent at a transmission rate $R_{is}=\frac{1}{n}H(W_{is})$ for $i=1,2$. 
%, where $\mathcal{W}_{i}=\{1,2,\dots,W_i\}$ is the messages set
The channel output are $Y^n_i\in \mathcal{Y}_i$ and $Z^n\in \mathcal{Z}$ at at user $i$ and the eavesdropper, respectively. 

For such a two-way wiretap channel, a $(2^{nR_{1s}},2^{nR_{2s}},n)$ code consists of:
\begin{itemize}
	\item Two independent message sets $\mathcal{W}_{1s}=\{1,2,\ldots,2^{nR_{1s}}\}$, $\mathcal{W}_{2s}=\{1,2,\ldots,2^{nR_{2s}}\}$.
	\item Two messages: $W_{1s}$ and $W_{2s}$ are independent and uniformly distributed over $\mathcal{W}_{1s}$ and $\mathcal{W}_{2s}$, respectively.
	\item Two encoders $f_1: \mathcal{W}_{1s} \rightarrow \mathcal{X}^n_1$, which map each message $w_{1s}\in \mathcal{W}_{1s}$ to a codeword $x^n_1\in \mathcal{X}^n_1$; $f_2: \mathcal{W}_{2s} \rightarrow \mathcal{X}^n_2$, which map each message $w_{2s}\in \mathcal{W}_{2s}$ to a codeword $x^n_2\in \mathcal{X}^n_2$.
	\item Two decoders $g_1: (\mathcal{Y}^n_1, \mathcal{X}^n_1) \rightarrow \hat{\mathcal{W}}_{2s}$, which map the received sequence $y^n_1$ and the sequence $x^n_1$ to a message $\hat{w}_{2s}$; $g_2: (\mathcal{Y}^n_2, \mathcal{X}^n_2) \rightarrow \hat{\mathcal{W}}_{1s}$, which map the received sequence $y^n_2$ and the sequence $x^n_2$ to a message $\hat{w}_{1s}$.
\end{itemize}

For a given code, two metrics should be sufficed: reliability and security. The reliability is measured by the {\em average error probabilities of decoding} at legitimate user 1 and 2, defined as 
\begin{align}
P_{e,1}=&\frac{1}{2^{nR_{2s}}} \sum_{W_{2s}=1}^{2^{nR_{2s}}} Pr\{\hat{W_{2s}}\neq W_{2s}\}; \nonumber\\
P_{e,2}=&\frac{1}{2^{nR_{1s}}} \sum_{W_{1s}=1}^{2^{nR_{1s}}} Pr\{\hat{W_{1s}}\neq W_{1s}\}. \label{Def_Error_Pr}
\end{align}
The individual security in this paper is defined by
\begin{align}
\frac{1}{n} I(W_{1s}; Z^n)  \leq \tau_n, \quad  \frac{1}{n} I(W_{2s}; Z^n)  \leq \tau_n,\quad \lim\limits_{n\to\infty} \tau_n=0,
\end{align}

\begin{Definition}\label{Def_Ach}
	The rate pair $(R_{1s}, R_{2s})$ is said to be achievable under the individual secrecy with $R_{1s}=\frac{1}{n}H(W_{1s}), \ R_{2s}=\frac{1}{n}H(W_{2s})$, if there exists a $(2^{nR_{1s}},2^{nR_{2s}},n)$ code such that
	\begin{align}
	& P_{e,i} \leq \epsilon_n, \quad\mbox{for}\quad i=1,2     \label{Eqn: Def_Reliability}\\
	&\frac{1}{n} I(W_{1s}; Z^n)  \leq \tau_n,  \quad \frac{1}{n} I(W_{2s}; Z^n)  \leq \tau_n,  \label{Eqn: Def_Individual_Secrecy}\\
	& \lim\limits_{n\to\infty} \epsilon_n=0 \quad \mbox{and}\quad \lim\limits_{n\to\infty} \tau_n=0.  \label{Eqn: Def_limits}
	\end{align}
\end{Definition}
Note that  \eqref{Eqn: Def_Reliability} indicates the reliability transmission constraint;   \eqref{Eqn: Def_Individual_Secrecy} is the individual secrecy constraint.

\begin{Remark}
	For the joint weak secrecy (joint secrecy for short in this paper) in \cite{tekin2008general,el2013achievable}, the rate of information leakage rate of both the messages $W_{1s}$ and $W_{2s}$ is demanded vanishing, i.e. 
	\begin{equation}\label{Eqn: Def_Joint Secrecy}
	\frac{1}{n} I(W_{1s}, W_{2s}; Z^n)  \leq \tau_n, \qquad  \lim\limits_{n\to\infty} \tau_n=0.  
	\end{equation}   
	
	If the coding schemes fulfill the  \eqref{Eqn: Def_Reliability}, \eqref{Eqn: Def_limits}, and the joint secrecy constraint  \eqref{Eqn: Def_Joint Secrecy}, then the rate pair $(R_{1s}, R_{2s})$ is said to be achievable under the joint secrecy constraint with $R_{1s}=\frac{1}{n}H(W_{1s}), \ R_{2s}=\frac{1}{n}H(W_{2s})$.		
\end{Remark}

%{\color{red}Note that we can rewrite the joint secrecy constraint as 
%	\begin{align*}
%	\frac{1}{n} I(W_{1s}, W_{2s}; Z^n)=\frac{1}{n} [I(W_{1s};Z^n)+I(W_{2s};Z^n|W_{1s})]  \leq \tau_n,  \qquad  \lim\limits_{n\to\infty} \tau_n=0.  
%	\end{align*}   
%	That is 
%	\begin{align}
%	\frac{1}{n} I(W_{1s};Z^n) \leq \tau_n,  \qquad \frac{1}{n} I(W_{2s};Z^n|W_{1s})  \leq \tau_n, \qquad  \lim\limits_{n\to\infty} \tau_n=0.  \label{Def_Joint_VS_Ind}
%	\end{align}   
%	Therefore, comparing the individual secrecy constraint with the joint secrecy constraint \eqref{Def_Joint_VS_Ind}, it is not sufficient to say that the joint secrecy constraint offers higher secrecy level than the individual secrecy in the two-way wiretap channel. }

However, the joint secrecy is not always affordable, such as the following lemma. 
	\begin{Lemma}
		Assume that in the two-way wiretap channels, the legitimate received symbol  $Y_1$ or $Y_2$ is a degraded version of the received symbol $Z$ at eavesdropper, i.e. $(X_1,X_2)\rightarrow Z\rightarrow Y_1$ or $(X_1,X_2)\rightarrow Z\rightarrow Y_2$ forms a Markov chain. Then, with the joint secrecy, the achievable secure transmission rates pair  $(R_{1s},R_{2s})$ with $R_{1s}>0,\ R_{2s}>0$ is not available, while with the individual secrecy constraint, it is available with the $R_{1s}>0,\ R_{2s}>0$.
	\end{Lemma}
	
	\begin{Proof}\label{Instance}	
		The information leakage of two messages $W_{1s}$ and $W_{2s}$ are
		\begin{align}
		H(W_{1s}, W_{2s}|Z^n)=& H(W_{1s}|Z^n)+H(W_{2s}|W_{1s}Z^n) \nonumber\\
		\stackrel{(a)}\leq & H(W_{1s}|Z^n)+H(W_{2s}|W_{1s}Y_1^n) \nonumber\\
		\stackrel{(b)}{\leq}& H(W_{1s}|Z^n)+n\epsilon_n  \nonumber\\
		\leq & H(W_{1s}) +n\epsilon_n  \nonumber\\
		=& nR_{1s}+n\epsilon_n \label{Equ: Eg_HW1W2}
		\end{align}
		where $(a)$ follows from the degraded assumption that  $Y_1$ is a degraded version of $Z$;  
		$(b)$ follows from the reliability transmission condition and Fano's inequality with $\lim\limits_{n\to\infty} \epsilon_n=0$.
		
		\begin{itemize}
			\item 	For the joint secrecy constraint 
			\begin{equation*}
			\frac{1}{n} I(W_{1s}, W_{2s}; Z^n)  \leq \tau_n, \qquad  \lim\limits_{n\to\infty} \tau_n=0,  
			\end{equation*}  
			we can obtain that 	
			\begin{align*}
			n(R_{1s}+R_{2s})=& H(W_{1s}, W_{2s})   \\
			=& H(W_{1s}, W_{2s}|Z^n)+I(W_{1s}, W_{2s};Z^n) \\
			\stackrel{(c)}\leq & H(W_{1s}, W_{2s}|Z^n)+n\tau_n   \\
			\stackrel{(d)} \leq & nR_{1s}+n\epsilon_n+n\tau_n 
			\end{align*}	
			where (d) follows from  the joint secrecy constraint; (d) follows from   \eqref{Equ: Eg_HW1W2}.
			That is 
			\begin{align}
			nR_{2s}\leq &n\epsilon_n+n\tau_n \label{Equ: Joint Spe_R_2}
			\end{align}		
			
			As $n$ goes to infinity (i.e. $n\to \infty$ ), according to \eqref{Eqn: Def_limits}, we have
			$\lim\limits_{n\to\infty} \epsilon_n=0, \lim\limits_{n\to\infty} \tau_n=0$. 
			Therefore, by \eqref{Equ: Joint Spe_R_2} we have
			\begin{align}
			R_{2s}\leq 0. \label{Equ: Eg_Joint_R2}
			\end{align}		
			
			Similarly, if $Y_2$ is a degraded version of $Z$, then $R_{1s}\leq 0$.		
			
			\item For the individual secrecy constraint
			\begin{equation}\label{Eqn: Def_Eg_individually secrecy}
			\frac{1}{n} I(W_{1s}; Z^n) \leq \tau_n, \qquad 
			\frac{1}{n} I(W_{2s}; Z^n) \leq \tau_n, \qquad \lim\limits_{n\to\infty} \tau_n=0.     
			\end{equation}  
			we have
			\begin{align*}
			n(R_{1s}+R_{2s})=& H(W_{1s}, W_{2s})   \\
			= & H(W_{1s}, W_{2s}|Z^n)+I(W_{1s}, W_{2s};Z^n)   \\
			\stackrel{(d)} \leq & nR_{1s}+n\epsilon_n+I(W_{1s};Z^n) + I(W_{2s};Z^n|W_{1s}) \\
			\stackrel{(e)}\leq & nR_{1s}+n\epsilon_n+n \tau_n + I(W_{2s};Z^n|W_{1s}) 
			\end{align*}
			where (d) follows from   \eqref{Equ: Eg_HW1W2}; (e) follows from the individual secrecy constraint  \eqref{Eqn: Def_Eg_individually secrecy}.
			
			As $n\to \infty$, according to  \eqref{Eqn: Def_limits}, that is 
			\begin{align}
			R_{2s}\leq I(W_{2s};Z^n|W_{1s}).  \label{Equ: Eg_Individual_R2}
			\end{align}
			Similarly, $R_{1s}\leq I(W_{1s};Z^n|W_{2s})$.
			Therefore, the individual secrecy can achieve both positive transmission rates pair $(R_{1s},R_{2s})$. 
			Moreover, if $Y_1$ is a degraded version of $Z$, we have $I(W_{2s};Z^n|W_{1s})\geq I(W_{2s};Y_1^n|W_{1s})$. This illustrates that $R_{2s}$ could even achieve $I(W_{2s};Y_1^n|W_{1s})$. 		
		\end{itemize}
	\end{Proof}
	
	In summary, if $Y_1$ or $Y_2$ is a degraded version of $Z$, the achievable secure transmission rates pair $(R_{1s}>0,R_{2s}>0)$ is not available under the joint secrecy constraint, while it is possible under the individual secrecy constraint. In the next section, we will give the exact achievable secrecy rate region of two-way wiretap channel with individual secrecy.

\section{Individual secrecy of two-way wiretap channel}\label{Sec: Individual}
In this section, we present our main results of two-way wiretap channel with individual secrecy. 
Firstly, we derive an achievable secrecy rate region of the general two-way wiretap channel, further give an intuitive interpretation of the result in a deterministic two-way wiretap channel. 
Secondly, we give an outer bound on the secrecy capacity.
%Both the inner and the outer bounds on the secrecy capacity will be explored in the binary input discrete two-way wiretap channels and the degraded Gaussian two-way wiretap channels in Section \ref{Sec: Example}.

\subsection{An achievable secrecy rate region}\label{Sec: Ind_Inner}
\begin{Theorem} \label{Thm_Inner_Individual}	
	For the two-way wiretap channels with an external eavesdropper, an achievable secrecy rate region is given by	 
	\begin{align*}
		\mathcal{R}^{Ind-In}\stackrel{\vartriangle}{=}\text{convex closure of}~\{ \bigcup_{p\in \mathcal{P}} \mathcal{R}^{Ind-In}(p)\}
	\end{align*} 
	where $\mathcal{P}$ denotes the set of all distribution of the random variables $U_1$, $U_2$, $X_1$, $X_2$ satisfying	
	$p(u_1u_2x_1x_2)=p(u_1)p(u_2)p(x_1|u_1)p(x_2|u_2)$;
	$\mathcal{R}^{Ind-In}(p)$ is the region of rate pairs $(R_{1s}, R_{2s})$ for $p\in \mathcal{P}$, satisfying
	\begin{equation}\label{Equ: Inner_Individual}
		\left\{		
		\begin{aligned}
			&(R_{1s}, R_{2s}): \\
			&R_{1s} \geq 0, R_{2s} \geq 0,  \\
			&R_{1s} \leq I(U_1; Y_2|X_2)-I(U_1;Z)-|I(U_2;Z|U_1)-I(U_2; Y_1|X_1)|^+, \\
			&R_{2s} \leq I(U_2; Y_1|X_1)-I(U_2;Z)-|I(U_1;Z|U_2)-I(U_1; Y_2|X_2)|^+.
		\end{aligned}
		\right\}	
	\end{equation}	
	and $|a|^+=\max\{0, a\}$, $|\mathcal{U}_1|\leq |\mathcal{X}_1|+1$, $|\mathcal{U}_2|\leq |\mathcal{X}_2|+1$.
\end{Theorem}

\begin{IEEEproof}
	See the proof in Appendix \ref{Sec: Pf_Thm_Inner_Individual}.
\end{IEEEproof}

Our achievable region is obtained by the stochastic encoding and the channel prefixing, where the codeword $U_1$ and $U_2$ are drawn from two binning codebooks respectively, and then passed on to two virtual prefix channel respectively. 
Accordingly, the channel input $X_1$ and $X_2$ are generated regarding to $p(x_1|u_1)$ and $p(x_2|u_2)$, respectively. 
Indeed, the channel prefixing is an interpretation of cooperative jamming \cite{ulukus2009cooperative}, which is a collaborative approach to improving the secrecy rate in a multi-user communication system. 
 
Specially, since $Z$ is related to $X_1$ and $X_2$ together, if the eavesdropper can decode part of message of user 2, it may help the eavesdropper to decode the confidential message $W_{1s}$. Hence, when analyzing the individual secrecy of $W_{1s}$, if $R_2\geq I(U_2;Z|U_1)$ then the codebook of user 2 is equally partitioned into $2^{R_{21}}$ sub-codebooks with $R_{21}=R_2-I(U_2;Z|U_1)+\epsilon'$, each part consisting of $2^{R_{22}}$ codewords with $R_{22}=I(U_2;Z|U_1)-\epsilon'$. The secrecy analysis of $W_{2s}$ works in a similar manner.

\begin{Remark}
It is worth noting that $R_{1s}$ and $R_{2s}$ meet the conditions individually in \eqref{Equ: Inner_Individual}. This phenomenon can be interpreted by considering the reliability and the individual secrecy of the system.
Firstly, for the reliability, the rate pair $(R_{1s}, R_{2s})$ should satisfy the achievable rate region given by Shannon in \cite{shannon1961two}, where $R_{1s}$ and $R_{2s}$ meet each condition separately with no trade-off between $R_{1s}$ and $R_{2s}$. 
Secondly, for the individual secrecy, $R_{1s}$ and $R_{2s}$ should meet $\frac{1}{n}I(W_{1s};Z)\leq \tau_n$ and $\frac{1}{n}I(W_{2s};Z)\leq \tau_n$, respectively. 
Therefore, as shown in Theorem \ref{Thm_Inner_Individual}, $R_{1s}$ and $R_{2s}$ are not directly interrelated with each other. 
Unlike the individual secrecy, the results in \cite{el2013achievable} revealed a trade-off between $R_{1s}$ and $R_{2s}$ with the joint secrecy. This is because the joint secrecy constraint $\frac{1}{n}I(W_{1s},W_{2s};Z)\leq \tau_n$ is related to both $R_{1s}$ and $R_{2s}$ at the same time. 
\end{Remark}

Applying Theorem \ref{Thm_Inner_Individual} to a two-way wiretap channel where the eavesdropper receives as many messages as the legitimate users, we have the following corollary. 
\begin{Corollary}\label{Corol_Individual_Bi}
	Suppose that in the two-way wiretap channel the legitimate receivers and the eavesdropper receive the same amount of messages, i.e. $Y_1=Y_2=Z$, an achievable secrecy rate region with individual secrecy is given by 
	\begin{align*}
		\mathcal{R}^{Ind-In}_1\stackrel{\vartriangle}{=}\text{convex closure of}~ \{ \bigcup_{p\in \mathcal{P}} \mathcal{R}^{Ind-In}_1(p)\},
	\end{align*} 
	where $\mathcal{P}$ denotes the set of all distribution of the random variables $X_1$, $X_2$ satisfying	$p(x_1x_2)=p(x_1)p(x_2)$.
	$\mathcal{R}^{Ind-In}_1(p)$ is the region of rate pairs $(R_{1s}, R_{2s})$ for $p\in \mathcal{P}$, satisfying
	\begin{equation}\label{Equ: Individual_Bi}
		\left\{		
		\begin{aligned}
			&(R_{1s}, R_{2s}): \\
			&R_{1s} \geq 0, R_{2s} \geq 0,  \\
			&R_{1s} \leq I(X_1; Y_2|X_2)-I(X_1; Z), \\
			&R_{2s} \leq I(X_2; Y_1|X_1)-I(X_2; Z).
		\end{aligned}
		\right\}	
	\end{equation}	
\end{Corollary} 
\begin{proof}
	Setting $U_1=X_1$ and $U_2=X_2$ in  \eqref{Equ: Inner_Individual}, then
	\begin{align*}
		R_{1s}& \leq I(X_1; Y_2|X_2)-I(X_1; Z)-|I(X_2;Z|X_1)-I(X_2; Z|X_1)|^+\\
		%& = I(X_1; Z|X_2)-I(X_1; Z)-I(X_2;Z|X_1)+R_2\\
		& = I(X_1; Y_2|X_2)-I(X_1; Z)
	\end{align*}
	Similarly, $R_{2s} \leq I(X_2; Y_1|X_1)-I(X_2; Z)$. This completes the proof.
\end{proof}

\begin{Remark}\label{Rem_Joint_in}
	With the joint secrecy, if $Y_1=Y_2=Z$ and $U_1=X_1$, $U_2=X_2$, the achievable secrecy rate region \cite{el2013achievable} is 
	\begin{align*}
	\mathcal{R}^{J-In}_1\stackrel{\vartriangle}{=}\text{convex closure of}~ \{ \bigcup_{p\in \mathcal{P}} \mathcal{R}^{J-In}_1(p)\},
	\end{align*} 
	where $\mathcal{P}$ denotes the set of all distribution of the random variables $X_1$, $X_2$ satisfying	$p(x_1x_2)=p(x_1)p(x_2)$.
	$\mathcal{R}^{J-In}_1(p)$ is the region of rate pairs $(R_{1s}, R_{2s})$ for $p\in \mathcal{P}$, satisfying
	\begin{equation}\label{Equ: Joint_Bi}
	\left\{		
	\begin{aligned}
	&(R_{1s}, R_{2s}): \\
	&R_{1s} \geq 0, R_{2s} \geq 0,  \\
	&R_{1s} \leq I(X_2; Y_1|X_1), \\
	&R_{2s} \leq I(X_1; Y_2|X_2), \\
	&R_{1s}+R_{2s} \leq I(X_2; Y_1|X_1)+I(X_1; Y_2|X_2)-I(X_1,X_2; Z).
	\end{aligned}
	\right\}	
	\end{equation}		
	
	Note that, the last equation can be rewritten as 
	\begin{align}
		R_{1s}+R_{2s} \leq & I(X_2; Y_1|X_1)+I(X_1; Y_2|X_2)-I(X_1,X_2; Z) \nonumber\\
		= & I(X_2; Y_1|X_1)-I(X_2; Z) \label{Equ_Joint_SumRate1} 
	\end{align}
	or	
	\begin{align}
		R_{1s}+R_{2s} \leq R_{1s}+R_{2s} \leq & I(X_1; Z|X_2)-I(X_1; Z) \label{Equ_Joint_SumRate2} 
	\end{align}	

	By comparing  \eqref{Equ: Individual_Bi} with \eqref{Equ_Joint_SumRate1} and \eqref{Equ_Joint_SumRate2}, either $R_{1s}$ or $R_{2s}$ with individual secrecy is equal to the sum-rate $R_{1s}+R_{2s}$ with the joint secrecy, which indicates that the secrecy rate region $\mathcal{R}^{J-In}_1$ with the joint secrecy is only half of the secrecy rate region $\mathcal{R}^{Ind-In}_1$ with the individual secrecy. 			
\end{Remark}

\subsection{An interpretation of Theorem \ref{Thm_Inner_Individual} in a deterministic two-way wiretap channel}\label{Sec: Interpretation}
		In this subsection, a deterministic two-way wiretap channel is studied to illustrate the intuition of the achievable secrecy rate region in Theorem \ref{Thm_Inner_Individual}.  		
		Suppose that the deterministic two-way wiretap channel is described by 
		\begin{align*}
		Y_1 =& X_1\oplus X_2\oplus N_1;\\
		Y_2 =& X_1\oplus X_2\oplus N_2;\\
		Z =& X_1\oplus X_2\oplus N_e;
		\end{align*}
		where $X_1, X_2, N_1, N_2, N_e\in \{0,1\}$; $X_1$ and $X_2$ are the binary channel inputs at user 1 and user 2, respectively; $Y_1$, $Y_2$ and $Z$ are the binary channel outputs at the user 1, user 2 and the eavesdropper, respectively; $N_1, N_2, N_e$ are the additive binary noise impairing user 1, user 2 and the eavesdropper, respectively.  
		Then, the corresponding transition probabilities are given by $p(N_1=1)=\varepsilon_1$, $p(N_2=1)=\varepsilon_2$ and $p(N_e=1)=\varepsilon_z$. 
		Therefore, the transmission probabilities are 
		\begin{align*}
		p(y_1\neq x_2|x_1)= & \varepsilon_1; \\
		p(y_2\neq x_1|x_2)= & \varepsilon_2; \\
		p(z\neq x_1\oplus x_2|x_1, x_2)= & \varepsilon_z.
		\end{align*}
		
		If the system does not have any eavesdropper, the model is reduced to a binary modulo-2 two-way channel that provides the reliability only. %According to the result given by Shannon in \cite{shannon1961two}, we obtain the following lemma for the reliable communication.
		\begin{Lemma}[\cite{shannon1961two}]\label{Lemma_Mod2_Reliable}
			For the full-duplex binary modulo-2 two-way channel, the achievable reliable transmit rate region $\mathcal{R}$ is the union of non-negative rate pairs $(R_{1}, R_{2})$ defined by
			\begin{align*}		
			R_{1} \leq & 1-h(\varepsilon_2), \\
			R_{2} \leq & 1-h(\varepsilon_1).
			\end{align*}
		\end{Lemma}	
		
		For the joint secrecy of the binary modulo-2 two-way wiretap channel, we have the following lemma.
		\begin{Lemma}[\cite{el2013achievable}]\label{Lemma_Mod2_Joint}
			For the full-duplex binary modulo-2  two-way wiretap channel with the joint secrecy, the achievable secrecy rate region $\mathcal{R}_s^{Joint-In}$ is the union of non-negative rate pairs $(R_{1s}, R_{2s})$ defined by
			\begin{align}		
			R_{1s} \leq & 1-h(\varepsilon_2), \nonumber\\
			R_{2s} \leq & 1-h(\varepsilon_1), \nonumber\\
			R_{1s}+R_{2s} \leq & 1+h(\varepsilon_z)-h(\varepsilon_1)-h(\varepsilon_2). \label{Equ: Eg_Joint_SumRate}
			\end{align}
		\end{Lemma}	 
		
		%On the other hand, according to our result in the following Section \ref{Sec: Individual}, Theorem \ref{Thm_Inner_Individual}, we obtain a corollary of the secrecy rate region with the individual secrecy for the modulo-2 binary two-way wiretap channel.
		On the other hand, according to Theorem 1,  the achievable secrecy rate region with the individual secrecy is given in the following corollary.
		\begin{Lemma}\label{Lemma_Mod2_Individual}	 
			For the full-duplex binary modulo-2 two-way wiretap channel with the individual secrecy, the achievable secrecy rate region $\mathcal{R}_s^{Ind-In}$ is the union of non-negative rate pairs $(R_{1s}, R_{2s})$ defined by	 
			\begin{align*}
			&R_{1s} \leq 1-h(\varepsilon_2), \\
			&R_{2s} \leq 1-h(\varepsilon_1).
			\end{align*} 	 
		\end{Lemma}
		\begin{IEEEproof}
			See the proof of Corollary \ref{Lemma_Mod2_Reliable}, \ref{Lemma_Mod2_Joint} and \ref{Lemma_Mod2_Individual} in Appendix \ref{Sec: Proof_Mod2_Example}.
		\end{IEEEproof}
		
		It is clear from Lemma \ref{Lemma_Mod2_Reliable}, \ref{Lemma_Mod2_Joint} and \ref{Lemma_Mod2_Individual} that the achievable rate region $\mathcal{R}$ is the same as $\mathcal{R}_s^{Ind-In}$, while $\mathcal{R}_s^{Joint-In}$ is smaller than $\mathcal{R}_s^{Ind-In}$ with regard to the sum-rate constraint  \eqref{Equ: Eg_Joint_SumRate}. We conclude the result in the following theorem. 				
		\begin{Theorem}\label{Theorem_Mod2}	 
			For the deterministic modulo-2 two-way wiretap channel, the achievable reliable transmit rate region $\mathcal{R}$, the achievable secrecy rate region $\mathcal{R}_s^{Joint-In}$ with the joint secrecy, and $\mathcal{R}_s^{Ind-In}$ with the individual secrecy satisfy
			\begin{align*}
				\mathcal{R}_s^{Joint-In}\subseteq \mathcal{R}_s^{Ind-In}=\mathcal{R}.
			\end{align*}
		\end{Theorem}
		
		From the practical viewpoint, Theorem \ref{Theorem_Mod2} reveals a great advantage of the individual secrecy that it can be achieved without any rate loss of the reliable transmission rate, yet there is a rate loss for the joint secrecy.  
		
		\begin{figure}[!htbp]		
			\subfloat[$h(\varepsilon_z)>h(\varepsilon_1),h(\varepsilon_2)$]{
				\begin{minipage}[t]{0.45\linewidth}
					\centering
					\includegraphics[width=0.8\textwidth]{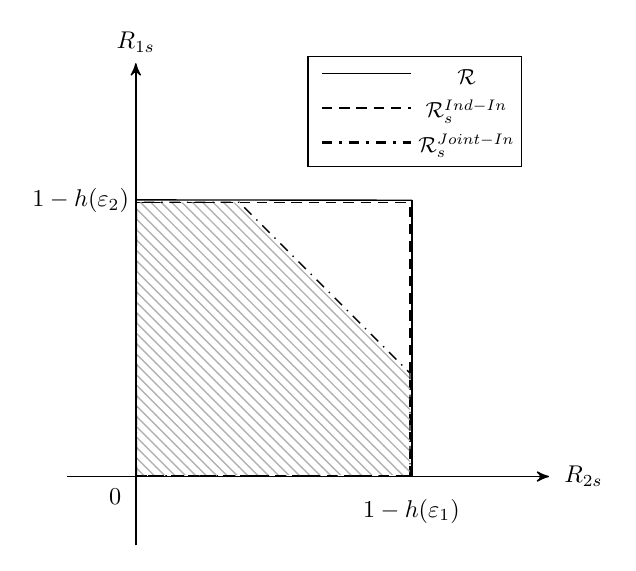}
					%\caption{(a).$h(\varepsilon_z)>h(\varepsilon_1),h(\varepsilon_2)$}	
			\end{minipage}}
			\hfill
			\subfloat[$h(\varepsilon_2)<h(\varepsilon_z)\leq h(\varepsilon_1)$]{\begin{minipage}[t]{0.45\linewidth}
					\centering
					\includegraphics[width=0.8\textwidth]{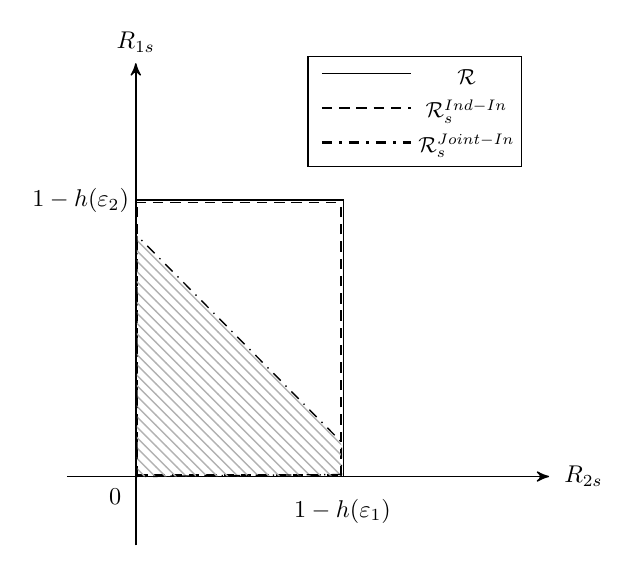}	
			\end{minipage} 	 }		
			
			\vfill
			\subfloat[$h(\varepsilon_1)<h(\varepsilon_z)\leq h(\varepsilon_2)$]{\begin{minipage}[t]{0.45\linewidth}
					\centering
					\includegraphics[width=0.8\textwidth]{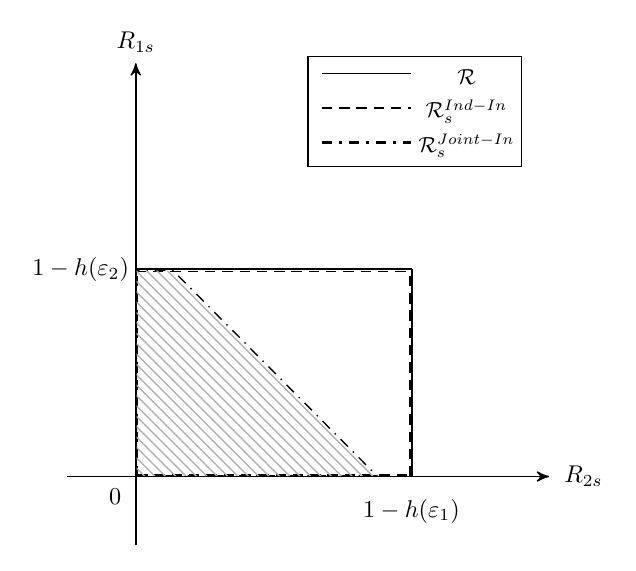}	
			\end{minipage} 	 }		
			\hfill
			\subfloat[$h(\varepsilon_z)<h(\varepsilon_1), h(\varepsilon_2)$]{\begin{minipage}[t]{0.45\linewidth}
					\centering
					\includegraphics[width=0.8\textwidth]{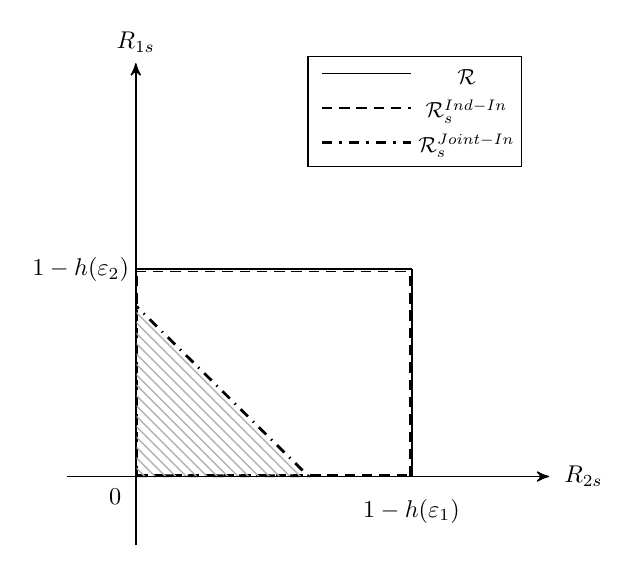}	
			\end{minipage} 	 }	 		
			\caption{$\mathcal{R}$, $\mathcal{R}_s^{Joint-In}$, $\mathcal{R}_s^{Indi-In}$ of binary modulo-2  two-way channel.}
			\label{Fig: Eg_Mod2}
		\end{figure}
		
		The geometric structures of $\mathcal{R}$, $\mathcal{R}_s^{Joint-In}$ and $\mathcal{R}_s^{Ind-In}$ are depicted by four cases regarding to the value of $h(\varepsilon_1), h(\varepsilon_2)$ and $h(\varepsilon_z)$ in Fig. \ref{Fig: Eg_Mod2}, where the boundary of $\mathcal{R}$, $\mathcal{R}_s^{Joint-In}$ and $\mathcal{R}_s^{Ind-In}$ are plotted by the solid line, the dashed-dotted line and the dashed line, respectively.
		In Fig. \ref{Fig: Eg_Mod2}, $\mathcal{R}$ coincides with $\mathcal{R}_s^{Ind-In}$ as a rectangle, and $\mathcal{R}_s^{Joint-In}$ contains a missing corner due to the constraint  $R_{1s}+R_{2s} \leq 1+h(\varepsilon_z)-h(\varepsilon_1)-h(\varepsilon_2)$ in Lemma \ref{Lemma_Mod2_Joint}. 
		Clearly, the individual secrecy provides a strictly larger secrecy rates region than the joint secrecy does, especially in the high rates region of $R_{1s}$ and $R_{2s}$.

\subsection{An outer bound on the secrecy capacity of two-way wiretap channels with individual secrecy}\label{Sec: Ind_Outer}
	\begin{Theorem} \label{Thm_Outer}	
		For the general two-way wiretapper channel, with the individual secrecy, an outer bound on the secrecy capacity is given by
		\begin{align*}
		\mathcal{R}^{Ind-O}\stackrel{\vartriangle}{=}\mathbf{Conv} \{ \bigcup_{p\in \mathcal{P}} \mathcal{R}^{Ind-O}(p)\},
		\end{align*} 	
		where $\mathcal{P}$ denotes the set of all distribution of the random variables $U$, $V$, $X_1$, $X_2$ satisfying		
		$p(uv_1v_2x_1x_2)=p(u)p(v_1|u)p(v_2|u)$ $p(x_1x_2|uv_1v_2)$; $U$, $V_1$ and $V_2$ are the auxiliary random variables and $U\to (V_1, V_2)\to (X_1, X_2)\to (Y_1, Y_2, Z)$ forms a Markov Chain, and $\mathcal{R}^{Ind-O}$ is the region of rate pairs $(R_{1s}, R_{2s})$ for $p\in \mathcal{P}$, satisfying
		\begin{equation}\label{Equ: Outer_R1R2_p}
		\left\{		
		\begin{aligned}
		&(R_{1s}, R_{2s}): \\
		&R_{1s} \geq 0, R_{2s} \geq 0,  \\
		&R_{1s} \leq I(V_1; X_{2}, Y_{2}|U)-I(V_1;Z|U), \\
		&R_{2s} \leq I(V_2; X_{1}, Y_{1}|U)-I(V_2;Z|U),
		\end{aligned}
		\right\}	
		\end{equation}
		 the cardinality of the auxiliary random variables $U$, $V_1$ and $V_2$ satisfies $|\mathcal{U}|\leq |\mathcal{X}_1||\mathcal{X}_2|+2$, $|\mathcal{V}_1|\leq (|\mathcal{X}_1||\mathcal{X}_2|+2)^2$ and $|\mathcal{V}_2|\leq (|\mathcal{X}_1||\mathcal{X}_2|+2)^2$.
	\end{Theorem}
	\begin{IEEEproof}
		See the proof in Appendix \ref{Proof_Outer}.
	\end{IEEEproof}
	
	The outer bound \eqref{Equ: Outer_R1R2_p} works for the general two-way wiretap channel with individual secrecy. Further, an outer bound derived for two classes of two-way channels in the following theorem.

\begin{Theorem} \label{Thm_Outer_ind_BiGTW}	
	For the following two classes of two-way wiretapper channels,
	\begin{enumerate}
		\item the legitimate users and the eavesdropper receive the same amount of messages, i.e. $Y_1=Y_2=Z=Y$;
		\item the received message $Z$ at the eavesdropper is a degraded version of both the messages at legitimate users, satisfying the Markov chain $Y_1\rightarrow Z$ and $Y_2\rightarrow Z$;
	\end{enumerate}
	an outer bound on the secrecy capacity is given by  
	\begin{align*}
		\mathcal{R}^{Ind-O}_{1}\stackrel{\vartriangle}{=} \{ \bigcup_{p\in \mathcal{P}} \mathcal{R}^{Ind-O}_{1}(p)\},
	\end{align*} 	
	where $\mathcal{P}$ denotes the set of all distribution of the random variables $Q$, $X_1$, $X_2$ with 	
	$p(qx_1x_2)$,
	and $\mathcal{R}^{Ind-O}(p)$ is the region of rate pairs $(R_{1s}, R_{2s})$ for $p\in \mathcal{P}$, satisfying
	\begin{equation}\label{Equ: Ind_outer_Bi}
		%\mathcal{R}^{ind-O}_{BTW}\stackrel{\vartriangle}{=}  \bigcup_{p(x_1x_2q)p(y_1,y_2,z|x_1,x_2)} 
		\left\{		
		\begin{aligned}
			&(R_{1s}, R_{2s}): \\
			&R_{1s} \geq 0, R_{2s} \geq 0,  \\
			&R_{1s} \leq I(X_1; Y_2|X_2,Q)-I(X_1;Z|Q), \\
			&R_{2s} \leq I(X_2; Y_1|X_1,Q)-I(X_2;Z|Q).
		\end{aligned}
		\right\}	
	\end{equation}	
	the cardinality of the auxiliary random variables $Q$ satisfies $|\mathcal{Q}|\leq |\mathcal{X}_1||\mathcal{X}_2|+1$.
\end{Theorem}
\begin{IEEEproof}
	See the proof in Appendix \ref{Sec: Pf_Ind_Outer_BGTW}.
\end{IEEEproof}
\begin{Remark}
	Later, Corollary \ref{Corol_Individual_Bi} and Theorem \ref{Thm_Outer_ind_BiGTW} will be applied into the binary input two-way wiretap channels to show the gap between the inner and the outer bound on the secrecy capacity.
	Specially, for a degraded Gaussian channel we will proof that the inner bound in Theorem \ref{Thm_Inner_Individual} and the outer bound in Theorem \ref{Thm_Outer_ind_BiGTW} coincide with each other, such that the secrecy capacity is fully established. 
\end{Remark}

\section{Binary input two-way wiretap channel and Degraded Gaussian two-way wiretap channel}\label{Sec: Example}
\subsection{Binary input two-way wiretap channel with individual secrecy} \label{Sec: Individual_Binary}
In this subsection, we are interested in the binary input two-way wiretap channels when the legitimate users and the eavesdropper have the same channel output, i.e. $Y_1=Y_2=Z$. 
As Shannon utilized the binary multiplying channel (BMC) to indicate the gap between the inner bound and the outer bound on the channel capacity of two-way channel, we also explore our main results in BMC to show the gap between the inner and the outer bound on the secrecy capacity.
Considering the binary-input (i.e., $x_1, x_2\in \{0,1\}$) and binary-output (i.e., $y_1=y_2=z\in \{0,1\}$) or ternary outputs (i.e., $y_1=y_2=z\in \{0,1,2\}$) transmission, the XOR channel and the Adder channel are also investigated.  

For each channel, the achievable secrecy rate region $\mathcal{R}^{Ind-In}$ and the outer bound on the secrecy capacity $\mathcal{R}^{Ind-O}$ are derived from Corollary \ref{Corol_Individual_Bi} and Theorem \ref{Thm_Outer_ind_BiGTW}, respectively.
%Correspondingly, $\mathcal{R}^{Ind-In}$ and $\mathcal{R}^{Ind-O}$ are plotted in the figures.
%, where the boundary is enclosed by the dash-dotted lines and the dashed lines, respectively.
%Additionally, the achievable secrecy rate region with the joint secrecy $\mathcal{R}^{J-In}$ is also plotted by the solid line as the boundary. 
For simplicity, we define the following operation
	 \begin{align*}
	 a\ast b :=a(1-b)+(1-a)b, \qquad \mbox{for} \quad 0\leq a, b\leq 1 
	 \end{align*}
 and the entropy function
	 \begin{equation}  
	 h(a):=\left\{  
	 \begin{array}{lr}  
	 -a\log a-(1-a)\log (1-a), & {\mbox{if} \ 0<a<1}\\
	 0,      & \ \ {\mbox{if} \ a=0\ \mbox{or}\ 1}.
	 \end{array}  
	 \right.  
	 \end{equation}  

\begin{figure}[!htbp]
	\centering
	\begin{tikzpicture}[node distance=2cm,auto,>=latex', scale=0.8, every node/.style={scale=0.85}]
	\draw (0.3,-0.3) -- (1,-1);
	\draw (1,-1) -- (3,-1);
	\draw (1,-1) -- (1,-3);
	\draw (3,-1) -- (3,-3);
	\draw (1,-3) -- (3,-3);
	\draw (1,-2) -- (3,-2);
	\draw (2,-1) -- (2,-3);
	\node at (0.8,-0.5) {$x_1$};
	\node at (0.5,-0.8) {$x_2$};
	\node at (1.5,-0.5) {$0$};
	\node at (2.5,-0.5) {$1$};
	\node at (0.5,-1.5) {$0$};
	\node at (0.5,-2.5) {$1$};
	
	\node at (1.5,-1.5) {$0$};
	\node at (2.5,-1.5) {$0$};
	\node at (1.5,-2.5) {$0$};
	\node at (2.5,-2.5) {$1$};
	\node at (2,-3.5) {$(a)$};
	
	\draw (3.8,-0.3) -- (4.5,-1);
	\draw (4.5,-1) -- (6.5,-1);
	\draw (4.5,-1) -- (4.5,-3);
	\draw (6.5,-1) -- (6.5,-3);
	\draw (4.5,-3) -- (6.5,-3);
	\draw (4.5,-2) -- (6.5,-2);
	\draw (5.5,-1) -- (5.5,-3);
	\node at (4.3,-0.5) {$x_1$};
	\node at (4,-0.8) {$x_2$};
	\node at (5,-0.5) {$0$};
	\node at (6,-0.5) {$1$};
	\node at (4,-1.5) {$0$};
	\node at (4,-2.5) {$1$};
	
	\node at (5,-1.5) {$0$};
	\node at (6,-1.5) {$1$};
	\node at (5,-2.5) {$1$};
	\node at (6,-2.5) {$0$};
	\node at (5.5,-3.5) {$(b)$};
	
	\draw (7.3,-0.3) -- (8,-1);
	\draw (8,-1) -- (10,-1);
	\draw (8,-1) -- (8,-3);
	\draw (10,-1) -- (10,-3);
	\draw (8,-3) -- (10,-3);
	\draw (8,-2) -- (10,-2);
	\draw (9,-1) -- (9,-3);
	\node at (7.8,-0.5) {$x_1$};
	\node at (7.5,-0.8) {$x_2$};
	\node at (8.5,-0.5) {$0$};
	\node at (9.5,-0.5) {$1$};
	\node at (7.5,-1.5) {$0$};
	\node at (7.5,-2.5) {$1$};
	
	\node at (8.5,-1.5) {$0$};
	\node at (9.5,-1.5) {$1$};
	\node at (8.5,-2.5) {$1$};
	\node at (9.5,-2.5) {$2$};
	\node at (9,-3.5) {$(c)$};
	
	\end{tikzpicture}
	\caption{Transition diagrams of the binary-input two-way channels.}
	\label{Fig_BiCase}
\end{figure}
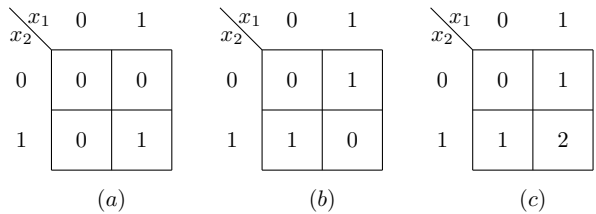

\subsubsection{Binary Multiplying channel} 

The BMC is shown in Fig. \ref{Fig_BiCase} (a), where the channel output is represented by $Y_1=Y_2=Z=X_1\cdot X_2$.
%The relation of channel input and output is shown in Fig. \ref{Fig_BiCase} (a), and it can be represented by $Y_1=Y_2=Z=X_1\cdot X_2$.
By Corollary \ref{Corol_Individual_Bi}, the achievable secrecy rate region $\mathcal{R}^{Ind-In}_{BMC}$ for BMC with individual secrecy is the union of the following non-negative rate pair $(R_{1s}, R_{2s})$ over $X_1\sim \mbox{Bern}(p_1), X_2\sim \mbox{Bern}(p_2)$:
	 \begin{align*}
	 &R_{1s} \leq p_2h(p_1)+p_1h(p_2)-h(p_1p_2), \\
	 &R_{2s} \leq p_2h(p_1)+p_1h(p_2)-h(p_1p_2).
	 \end{align*}

\begin{figure}[!htbp]
	\centering
	\includegraphics[width=0.4\textwidth]{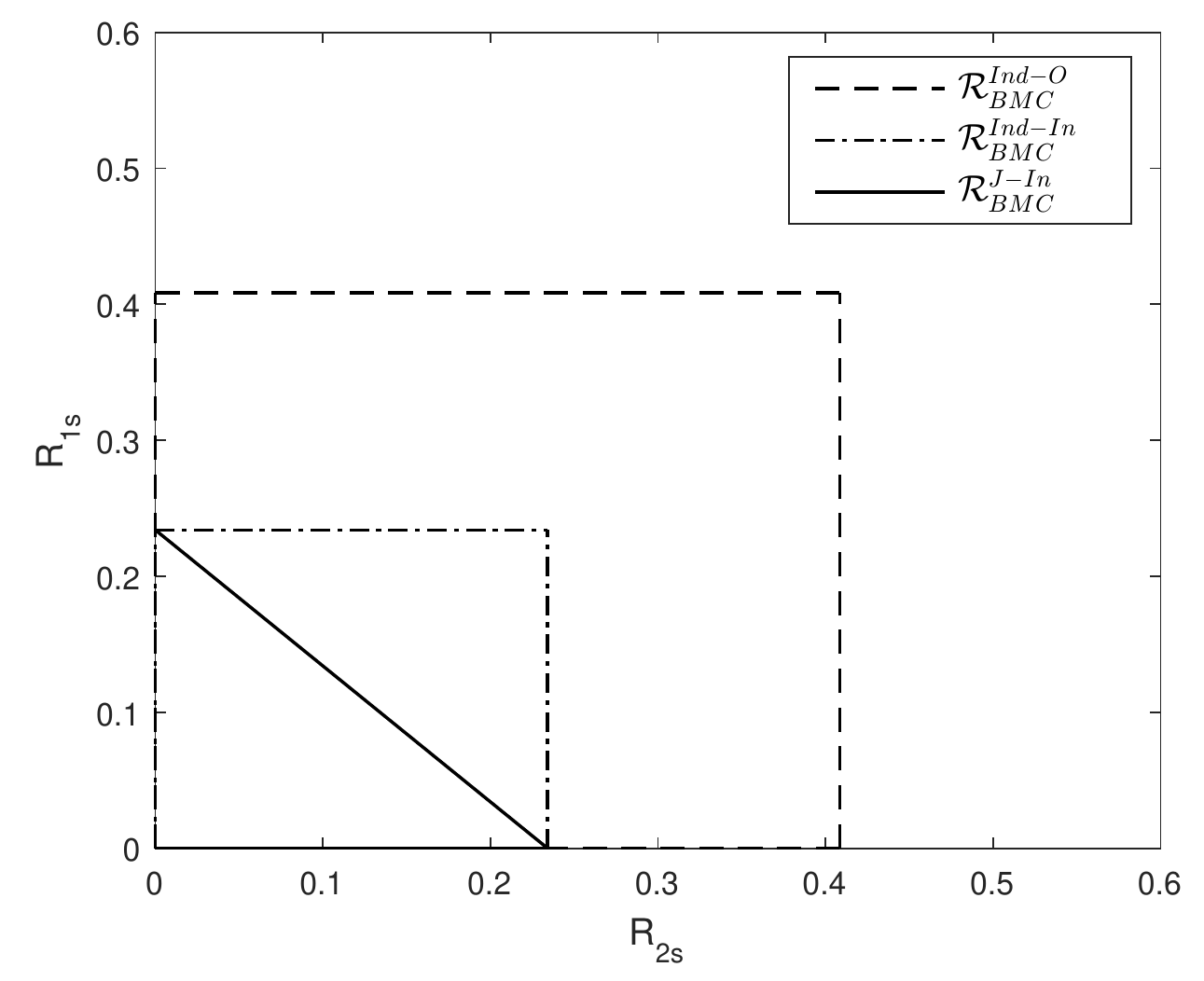}
	\caption{Secrecy rate region of BMC channel.}
	\label{Fig: Bi_Ind_BMC}
\end{figure}

The achievable secrecy rate region $\mathcal{R}^{Ind-In}_{BMC}$ and the outer bound $\mathcal{R}^{Ind-O}_{BMC}$ are shown in Fig. \ref{Fig: Bi_Ind_BMC}. Additionally, the achievable secrecy rate region with the joint secrecy $\mathcal{R}^{J-In}_{BMC}$ is also plotted for comparison. 
%From the knowledge of the BMC (with no secrecy constraint) \cite{shannon1961two}, .
The numerical results in Fig. \ref{Fig: Bi_Ind_BMC} demonstrate that the region $\mathcal{R}^{Ind-In}_{BMC}$ is twice as large as $\mathcal{R}^{J-In}_{BMC}$, consistent with Remark \ref{Rem_Joint_in}. 
Moreover, it can be seen that the increase in $R_{2s}$ leads to the decrease in $R_{1s}$ on $\mathcal{R}^{J-In}_{BMC}$, while $R_{1s}$ and $R_{2s}$ are greatly improved and achieve high secrecy rate simultaneously on $\mathcal{R}^{Ind-In}_{BMC}$. 
However, the gap between $\mathcal{R}^{Ind-In}_{BMC}$ and $\mathcal{R}^{Ind-O}_{BMC}$ is still large. 

\subsubsection{Binary XOR channel} 

The XOR channel is shown in Fig. \ref{Fig_BiCase} (b), where the channel output is represented by $Y_1=Y_2=Z=X_1\oplus X_2$.
By Corollary \ref{Corol_Individual_Bi}, the achievable secrecy rate region $\mathcal{R}^{Ind-In}_{XOR}$ for XOR with individual secrecy is the union of the following non-negative rate pair $(R_{1s}, R_{2s})$ over $X_1\sim \mbox{Bern}(p_1), X_2\sim \mbox{Bern}(p_2)$:
	 \begin{align*}
	 R_{1s} &\leq h(p_1)+h(p_2)-h(p_1*p_2), \\	
	 R_{2s} &\leq h(p_1)+h(p_2)-h(p_1*p_2).
	 \end{align*}	 

\begin{figure}[!htbp]
	\centering
	\includegraphics[width=0.45\textwidth]{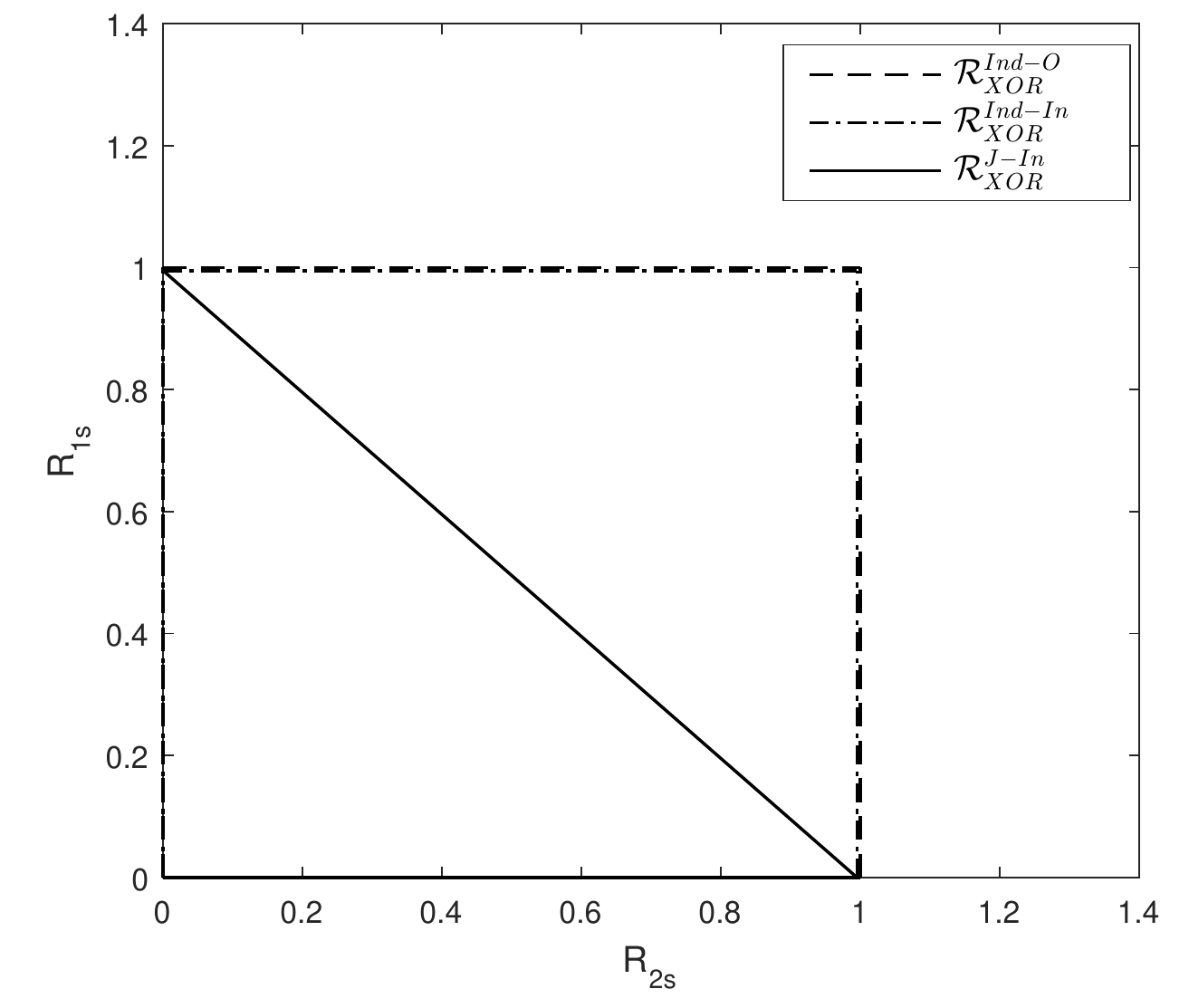}
	\caption{Secrecy rate region of XOR channel.}
	\label{Fig: Bi_Ind_XOR}
\end{figure}

Correspondingly, the achievable secrecy rate region with individual secrecy $\mathcal{R}^{Ind-In}_{XOR}$, with the joint secrecy $\mathcal{R}^{J-In}_{XOR}$ and the outer bound $\mathcal{R}^{Ind-O}_{XOR}$ with individual secrecy are shown in Fig. \ref{Fig: Bi_Ind_XOR}.  
Clearly, $\mathcal{R}^{J-In}_{XOR}$ is only half the size of $\mathcal{R}^{Ind-In}_{XOR}$.
Especially, the maximum achievable secrecy rate on $\mathcal{R}^{Ind-In}_{XOR}$ is $(R_{1s}, R_{2s})=(1, 1)$, which is also the maximum reliable rate without secrecy \cite{shannon1961two}. It indicates that the individual secrecy can be achieved with no rate loss of reliable transmission.
Moreover, $\mathcal{R}^{Ind-I}_{XOR}$ coincides with $\mathcal{R}^{Ind-O}_{XOR}$, hence the individual secrecy capacity region of XOR channel is fully characterized with $(R_{1s}\leq 1, R_{2s}\leq 1)$.

\subsubsection{Adder channel} 

The XOR channel is shown in Fig. \ref{Fig_BiCase} (c), where the channel output is represented by $Y_1=Y_2=Z=X_1+X_2$.
 By Corollary \ref{Corol_Individual_Bi}, the achievable secrecy rate region $\mathcal{R}^{Ind-In}_{Adder}$ for the binary Adder channel with individual secrecy is the union of the following non-negative rate pair $(R_{1s}, R_{2s})$ over $X_1\sim \mbox{Bern}(p_1), X_2\sim \mbox{Bern}(p_2)$:
	 \begin{align*}
	 R_{1s} &\leq (p_1*p_2)h\left(\frac{p_1(1-p_2)}{p_1*p_2}\right),\\
	 R_{2s} &\leq (p_1*p_2)h\left(\frac{p_2(1-p_1)}{p_1*p_2}\right).  
	 \end{align*}

\begin{figure}[!htbp]
	\centering
	\includegraphics[width=0.45\textwidth]{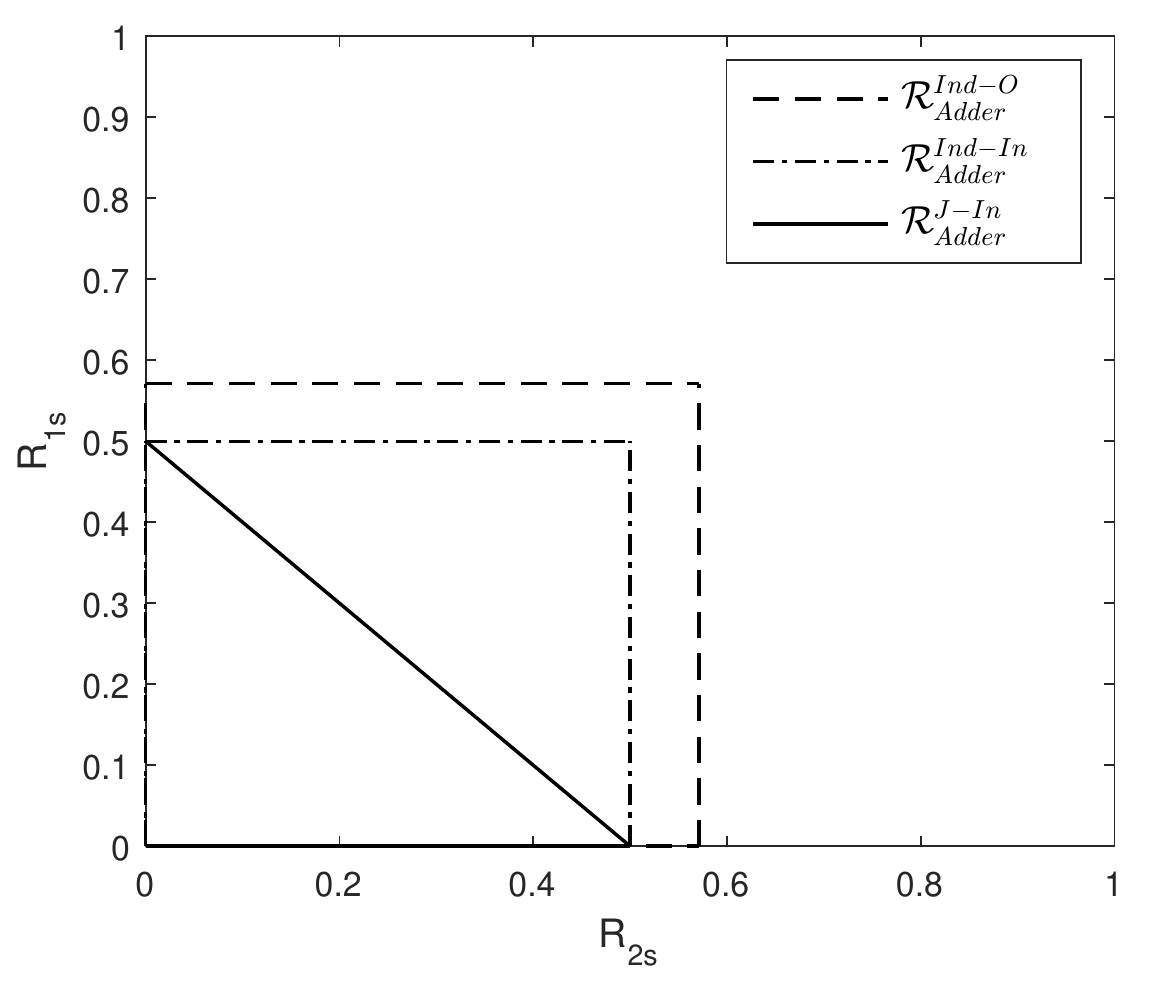}
	\caption{Secrecy rate region of Adder channel.}
	\label{Fig: Bi_Ind_Add}
\end{figure} 

Correspondingly, with individual secrecy, the achievable secrecy rate region $\mathcal{R}^{Ind-In}_{Adder}$ and the outer bound $\mathcal{R}^{Ind-O}_{Adder}$ are drawn in Fig. \ref{Fig: Bi_Ind_Add}, where the achievable secrecy rate region with the joint secrecy $\mathcal{R}^{J-In}_{Adder}$ is also plotted for comparison. 
As in BMC and XOR channel, in Adder channel $\mathcal{R}^{Ind-In}_{Adder}$ is twice as large as $\mathcal{R}^{J-In}_{Adder}$.
Moreover, for the individual secrecy, the gap between $\mathcal{R}^{Ind-I}_{Adder}$ and $\mathcal{R}^{Ind-O}_{Adder}$ has narrowed considerably than that of the BMC.

\subsection{Degraded Gaussian two-way wiretap channel with individual secrecy} \label{Sec: Individual_Gaussian}
In this subsection, we study a class of degraded Gaussian two-way wiretap channels with individual secrecy. 
We first define two classes of degraded channels.

In the two-way wiretap channel, suppose the channel inputs of the two users are $x_1$ and $x_2$, respectively; the channel output at the users and the eavesdropper are $y_1$, $y_2$ and $z$, respectively. 
\begin{Definition}
	The two-way wiretap channel is physically degraded if the transition probability distribution satisfies
	\begin{align}
	p(z,y_1|x_1,x_2)=p(y_1|x_1,x_2)p(z|y_1). \label{Def_Phy_Deg}
	\end{align}	
\end{Definition}
\begin{Definition}
	The two-way wiretap channel is physically degraded if the conditional marginal distribution is the same as that of a physically degraded two-way wiretap channel, i.e., there exists a distribution $p(z|y_1)$ such that 
	\begin{align}
	p(z|x_1,x_2)=\sum_{y_1}p(y_1|x_1,x_2)p(z|y_1). \label{Def_Sto_Deg}
	\end{align}	
\end{Definition} 

Assume that the channel is discrete and memoryless, and the channel outputs at the legitimate receivers and the eavesdropper are corrupted by additive Gaussian noise terms.
Then, the channel outputs at each time $i$ are given by
\begin{subequations}
	\begin{align}
	Y_{1i}=& X_{1i}+X_{2i}+Z_{1i}; \\
	Y_{2i}=& X_{1i}+X_{2i}+Z_{2i}; \\
	Z_{i}=&  X_{1i}+X_{2i}+Z_{ei}; \label{Equ_Deg_Gau}
	\end{align}
\end{subequations}
where $Z_{1i}$, $Z_{2i}$ and $Z_{ei}$ are independent zero-meaning additive Gaussian noises with $Z_{1i}\sim \mathcal{N}\{0, N_1\}$, $Z_{2i}\sim \mathcal{N}\{0, N_2\}$, $Z_{ei}\sim \mathcal{N}\{0, N_e\}$, and $N_e>N_1,\ N_e>N_2$. The average power constraints of the channel input sequences $X_1^n$ and $X_2^n$ are 
\begin{equation} \label{Equ: PowerCons}
\dfrac{1}{n} \sum_{i=1}^{n} E[X_{1i}^2]\leq P_1 \quad \text{and} \quad  \dfrac{1}{n} \sum_{i=1}^{n} E[X_{2i}^2]\leq P_2.
\end{equation}

Under the assumption $N_e>N_1$, the output $Z$ is a stochastically degraded version of $Y_1$, since the marginal distribution $p(z|x_1,x_2)$ is the same as that of the following physical degraded Gaussian two-way channel:
\begin{align}
Y_{1i}=& X_{1i}+X_{2i}+Z_{1i}; \\ 		
Z_{i}=&  X_{1i}+X_{2i}+Z_{1i}+Z'_{1i};
\end{align}	
where $Z'_{1i}$ is independent of $Z_{1i}$, being zero-meaning Gaussian noises with variance $Z_{ei}-Z_{1i}$. Similarly, under the assumption $N_e>N_2$, the output $Z$ is a stochastically degraded version of $Y_2$.

We have the fundamental limits of the degraded Gaussian two-way wiretap channel with individual secrecy in the following theorem. 
	\begin{Theorem} \label{Thm_Capacity_DegGau}
		For the degraded Gaussian two-way wiretap channel with the individual secrecy, if the received symbol $Z$ at eavesdropper is a stochastically degraded version of the received symbol $Y_1$ and $Y_2$ at the legitimate users, the secrecy capacity is given by
		\begin{equation}\label{Equ:Individual_Capacity_Gau}
		\mathcal{R}^{Ind-c}_{GTW}\stackrel{\vartriangle}{=}  
		\left\{		
		\begin{aligned}
		&(R_{1s}, R_{2s}): \\
		&R_{1s} \geq 0, R_{2s} \geq 0,  \\
		&R_{1s}\leq \frac{1}{2}\log \frac{(P_1+N_2)(P_2+N_e)}{N_2(P_1+P_2+N_e)}, \\
		&R_{2s} \leq \frac{1}{2}\log \frac{(P_2+N_1)(P_1+N_e)}{N_1(P_1+P_2+N_e)}.
		\end{aligned}
		\right\}	
		\end{equation}	
	\end{Theorem}
	
	\begin{IEEEproof}
		See Appendix \ref{Sec: Proof_Def_Gau}. 
	\end{IEEEproof}
	
	\begin{figure} [!htbp]
		\centering
		\includegraphics[width=0.45\textwidth]{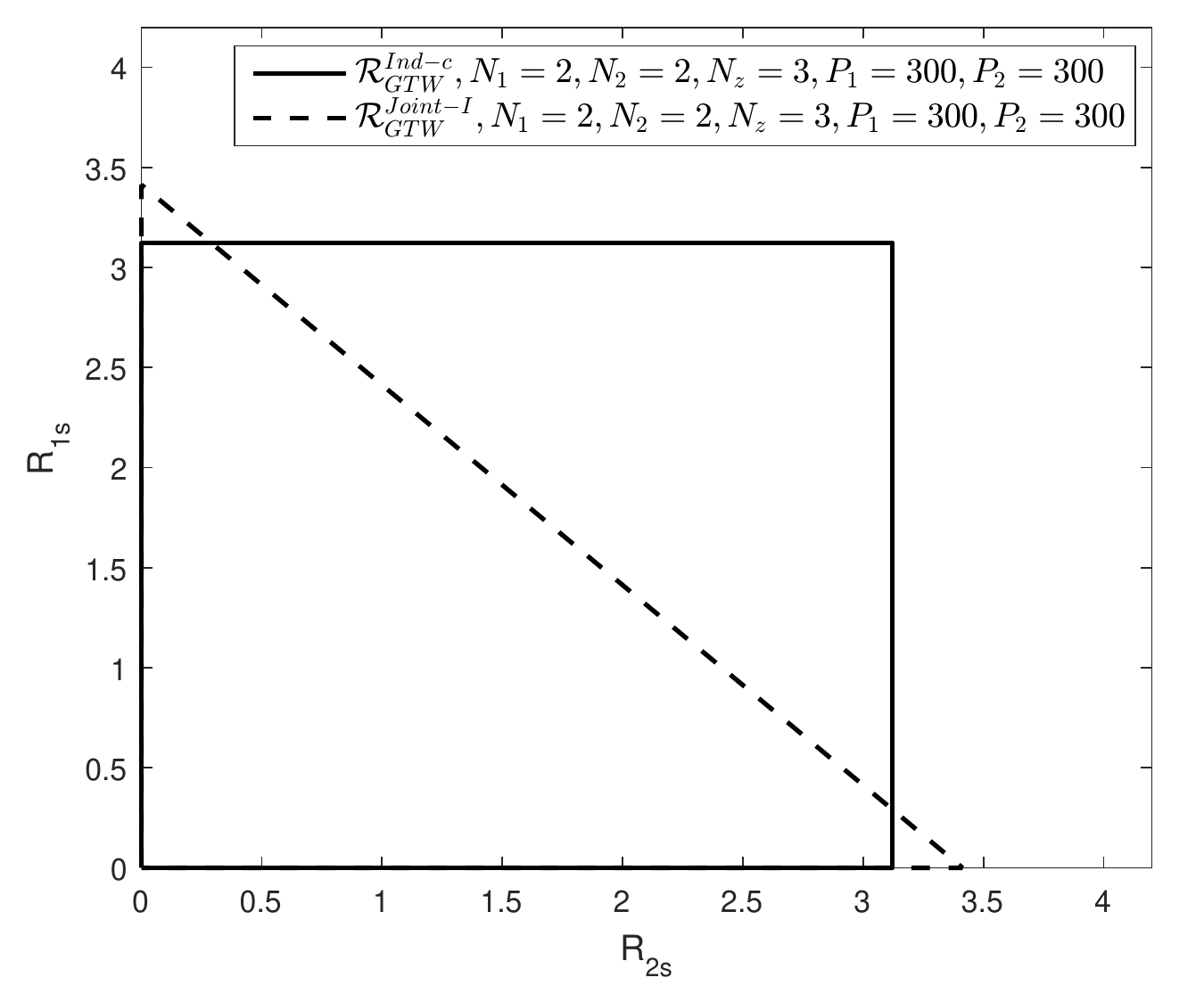}
		\caption{the secrecy capacity of the degraded Gaussian two-way wiretap channel with the individual secrecy and the achievable secrecy rate region with the joint secrecy, with $N_1=2,\ N_2=2,\ N_e=3,\ P_1=300,\ P_2=300$.}
		\label{fig: Individual_Gau}
	\end{figure}
	
	In Fig. \ref{fig: Individual_Gau}, we plot $\mathcal{R}^{Ind-c}_{GTW}$ and the achievable secrecy rate region $\mathcal{R}^{Joint-I}_{GTW}$ with the joint secrecy for $N_1=2, N_2=2, N_e=3, P_1=300, P_2=300$.
	Firstly, on $\mathcal{R}^{Ind-c}_{GTW}$, $R_{1s}$ and $R_{2s}$ achieve high rate region simultaneously, i.e. $R_{1s}=3.1228,R_{1s}=3.1228$. However, $R_{1s}$ and $R_{2s}$ can not achieve such high rate simultaneously as on $\mathcal{R}^{Joint-I}_{GTW}$, where if $R_{1s}$ is as high as $R_{1s}=3.1228$, then the $R_{2s}$ is only $R_{2s}=0.2901$.
	Secondly, in lower secrecy rate region $R_{2s}<0.2901$, the joint secrecy offers higher secrecy rate $R_{1s}$ than the individual secrecy does. This is because that two confidential messages are guaranteed secure under the joint secrecy constraint $\frac{1}{n} I(W_{1s}, W_{2s}; Z^n)  \leq \tau_n$ which is referred to $R_{1s}$ and $R_{1s}$ together. By sacrificing the rate $R_{2s}$, the secrecy rate $R_{1s}$ can be improved, even achieving the $(R_{1s})_{max}=3.4129, R_{2s}=0$.
	Nevertheless, the individual secrecy provides higher secrecy sum-rate $R_{1s}+R_{2s}$ than the joint secrecy does. From Fig. \ref{fig: Individual_Gau}, it can be seen $(R_{1s}+R_{2s})_{Ind}=6.2456$ for the individual secrecy and $(R_{1s}+R_{2s})_{Joint}=3.4129$ for the joint secrecy, which indicates that the sum-rate $R_{1s}+R_{2s}$ with individual secrecy is much larger than that with joint secrecy.  	
	Actually, the maximum sum-rate $R_{1s}+R_{2s}$ on $\mathcal{R}^{Joint-I}_{GTW}$ is
		\begin{align}		
		(R_{1s}+R_{2s})_{Joint}\leq & \frac{1}{2}\log \frac{(P_1+N_2)(P_2+N_1)N_e}{N_2N_1(P_1+P_2+N_e)}. \label{Equ_DeGau_J_SumR}
		\end{align}
	By \eqref{Equ:Individual_Capacity_Gau}, the sum-rate $R_{1s}+R_{2s}$ with the individual secrecy satisfies
		\begin{align}		
		(R_{1s}+R_{2s})_{Ind}\leq & \frac{1}{2}\log \frac{(P_1+N_2)(P_2+N_1)(P_1+N_e)(P_2+N_e)}{N_1N_2(P_1+P_2+N_e)(P_1+P_2+N_e)}.\label{Equ_DeGau_Ind_SumR}
		\end{align}
	Such that
		\begin{align*}		
		(R_{1s}+R_{2s})_{Joint}-(R_{1s}+R_{2s})_{Ind}=& \frac{1}{2}\log \frac{N_e(P_1+P_2+N_e)}{(P_1+N_e)(P_2+N_e)} \\
		%=& \frac{n}{2}\log \frac{N_e}{(P_1+N_e)} \times \frac{(P_1+P_2+N_e)}{(P_2+N_e)} \\
		=&\frac{1}{2}\log  \frac{1}{1+\frac{P_1}{N_e}}\times  (1+\frac{P_1}{P_2+N_e})\\
		<& 0
		\end{align*}	
	Theoretically, the sum-rate $R_{1s}+R_{2s}$ with the individual secrecy is strictly larger than that with the joint secrecy, hence it is consistent with the numerical results in Fig. \ref{fig: Individual_Gau}.

\section{Conclusion}\label{Sec: Conclusion}
In this paper, we investigated the fundamental limits of two-way wiretap channel with individual secrecy. 
Firstly, by channel prefixing approach and stochastic encoding, we derived an achievable secrecy rate region for the general two-way wiretap channel. Secondly, we obtained outer bounds on the secrecy capacity region for the general two-way wiretap channel and for two classes of special two-way wiretap channels. 
The result showed that the individual secrecy creates an advantage over the joint secrecy for the achievable secrecy rate region in a binary modulo-2 two-way channel, where the region with individual secrecy was shown to be twice as large as that with joint secrecy. 
Particularly, the inner and the outer bound coincide with each other in XOR channel and degraded Gaussian two-way wiretap channel, hence the secrecy capacity regions were established. 
In addition, in the degraded Gaussian two-way wiretap channel, the individual secrecy gains larger secure sum-rate than the joint secrecy does.   	
\appendices

\section{Proof of Lemma \ref{Lemma_Mod2_Reliable}, \ref{Lemma_Mod2_Joint} and \ref{Lemma_Mod2_Individual}}\label{Sec: Proof_Mod2_Example}
In order to give the achievable reliable transmission rate region $\mathcal{R}$ and the secrecy rate region $\mathcal{R}_s^{Joint-In}$ and $\mathcal{R}_s^{Indi-In}$, we first calculate the following terms.
\begin{itemize}
	\item $I(X_1;Y_2|X_2)$ and $I(X_2;Y_1|X_1)$:
		\begin{align*}
		& I(X_1;Y_2|X_2) \\
		= & H(Y_2|X_2)-H(Y_2|X_1,X_2) \\
		\leq & 1-\sum_{x_2} p(x_2)H(Y_2|X_1,x_2) \\
		= & 1-h(y_2\neq x_1|x_2) \\
		= & 1-h( \varepsilon_2)
		\end{align*}
		Similarly, we have
		\begin{align*}
		I(X_2;Y_1|X_1)\leq & 1-h( \varepsilon_1) ;
		\end{align*}		
	\item $I(X_1;Y_2|X_2)+I(X_2;Y_1|X_1)-I(X_1,X_2;Z)$:
		\begin{align*}
		& I(X_1;Y_2|X_2)+I(X_2;Y_1|X_1)-I(X_1,X_2;Z) \\
		= & H(Y_2|X_2) +H(Y_1|X_1)-H(Z)-H(Y_2|X_1X_2)-H(Y_1|X_1X_2)+H(Z|X_1X_2)
		\end{align*}
		By noting that,
		\begin{align} 
		& H(Y_2|X_2) +H(Y_1|X_1)-H(Z) \notag\\
		=& H(X_1\oplus N_2)+H(X_2\oplus N_1)-H(X_1\oplus X_2\oplus N_e) \notag\\
		\stackrel{(a)}=& H(X_1\oplus N_2)+H(X_2\oplus N_1)-H(X_1\oplus N_2\oplus X_2\oplus N_1 \oplus \hat{N_e}) \notag\\
		\stackrel{(b)} \leq & H(X_1\oplus N_2)+H(X_2\oplus N_1)-H(X_1\oplus N_2\oplus X_2\oplus N_1) \notag\\
		=& H(X_1\oplus N_2)+H(X_1\oplus N_2\oplus X_2\oplus N_1|X_1\oplus N_2)-H(X_1\oplus N_2\oplus X_2\oplus N_1) \notag\\
		=& H(X_1\oplus N_2)-I(X_1\oplus N_2\oplus X_2\oplus N_1;X_1\oplus N_2)  \notag\\
		=& H(X_1\oplus N_2 |X_1\oplus N_2\oplus X_2\oplus N_1) \notag\\
		\leq & 1 \label{Equ_Eg_Mod2_R12s}
		\end{align}

		where $(a)$ follows by setting $\hat{N_e}=N_1 \oplus N_2\oplus N_e$;  $(b)$	follows from the fact that conditioning does not increase entropy.
		
		Hence, we conclude that
		\begin{align*}
		& I(X_1;Y_2|X_2)+I(X_2;Y_1|X_1)-I(X_1,X_2;Z) \\
		\stackrel{(c)}\leq & 1-H(Y_2|X_1X_2)-H(Y_1|X_1X_2)+H(Z|X_1X_2) \\
		\stackrel{(d)}=& 1-h(\varepsilon_1)-h(\varepsilon_2)+h(\varepsilon_z)
		\end{align*}		
		where $(c)$ follows from  \eqref{Equ_Eg_Mod2_R12s}; $(d)$ follows from $H(Y_2|X_1X_2)=h(\varepsilon_2)$, $H(Y_1|X_1X_2)=h(\varepsilon_1)$, and $H(Z|X_1X_2)=h(\varepsilon_z)$. 
	\item $I(X_1;Y_2|X_2)-I(X_1;Z)$ and $I(X_1;Y_2|X_2)-I(X_1;Z)$:
		\begin{align*}
		& I(X_1;Y_2|X_2)-I(X_1;Z) \\
		= & I(X_1;Y_2|X_2)+I(X_2;Z|X_1)-I(X_1,X_2;Z) \\
		= & H(Y_2|X_2) +H(Z|X_1)-H(Z)-H(Y_2|X_1X_2)-H(Z|X_1X_2)+H(Z|X_1X_2) \\
		= & H(Y_2|X_2) +H(Z|X_1)-H(Z)-H(Y_2|X_1X_2) \\
		\stackrel{(e)}\leq & 1-h(\varepsilon_2)
		\end{align*}		
		where $(e)$ follows from the same process as  \eqref{Equ_Eg_Mod2_R12s}, and $H(Y_2|X_1X_2)=h(\varepsilon_2)$.
		
		Similarly, we have 
		\begin{align*}
		I(X_2;Y_1|X_1)-I(X_2;Z) \leq & 1-h(\varepsilon_1).
		\end{align*}
\end{itemize}

Based on these items above, the reliable transmission rate given by Shannon \cite{shannon1961two}, the achievable reliable transmission rate pair $(R_1,R_2)$ satisfies
	\begin{align*}
	R_{1s} \leq & I(X_1;Y_2|X_2)  \leq  1-h(\varepsilon_2), \\
	R_{2s} \leq & I(X_2;Y_1|X_1) \leq 1-h(\varepsilon_1).
	\end{align*} 	

According to the results in previous work \cite{el2013achievable}, we can obtain the achievable rate pair $(R_{1s}, R_{2s})$ for the deterministic modulo-2 two-way wiretap channel with the joint secrecy as 
	\begin{align*}		
	R_{1s} \leq & I(X_1;Y_2|X_2)-I(X_1;Z) \leq 1-h(\varepsilon_2), \\
	R_{2s} \leq & I(X_2;Y_1|X_1)-I(X_2;Z) \leq 1-h(\varepsilon_1), \\
	R_{1s}+R_{2s} \leq & I(X_1;Y_2|X_2)+I(X_2;Y_1|X_1)-I(X_1,X_2;Z) \leq 1+h(\varepsilon_z)-h(\varepsilon_1)-h(\varepsilon_2).
	\end{align*}

%And, according to our result in Theorem \ref{Thm_Inner_Individual}, the secrecy rate $(R_{1s}, R_{2s})$ with the individual secrecy for the modulo-2 binary two-way wiretap channel is 
By Theorem \ref{Thm_Inner_Individual}, the secrecy rate pair $(R_{1s}, R_{2s})$ with the individual secrecy for the modulo-2 binary two-way wiretap channel satisfies
	\begin{align*}
	R_{1s} \leq &  I(X_2;Y_1|X_1)-I(X_2;Z)\leq  1-h(\varepsilon_2), \\
	R_{2s} \leq & I(X_1;Y_2|X_2)-I(X_1;Z)\leq 1-h(\varepsilon_1).
	\end{align*} 	 

%\section{Proof of Theorem \ref{Thm_Inner_Bd}}\label{sec: Proof of Theorem 1}
%\section{Proof of Theorem \ref{Thm_Inner_Individual}}\label{Sec: Proof_Thm_Individual}

\section{Proof of Theorem \ref{Thm_Inner_Individual}}\label{Sec: Pf_Thm_Inner_Individual}
%Assume that the encodings at both legitimate users are conducted independently.
With fixed probability density function $p(u_1)$ and $p(u_2)$, the codebooks are generated as follows.

	\subsubsection{Codebook generation}
	According to $p(u_i)$, the user $i$, ($i=1,2$), randomly generates $2^{nR_{i}}$ independent and identically distributed (i.i.d) sequences $u^n_i(w_{is}, w_{ir})$, with $(w_{is}, w_{ir})\in [1: 2^{nR_{is}}]\times [1:2^{nR_{ir}}]$.  Note that
		\begin{equation} \label{Eqn: Relation R1 R1s}
			R_i=R_{is}+R_{ir}.
		\end{equation}  	

	When analysing the secrecy measurement $\frac{1}{n}H(W_{1s}|Z^n)$, if $R_2\geq I(U_2;Z|U_1)$ then the codebook of user 2 is equally partitioned into $2^{R_{21}}$ parts with $R_{21}=R_2-I(U_2;Z|U_1)+\epsilon'$, each part consisting of $2^{R_{22}}$ codewords with $R_{22}=I(U_2;Z|U_1)-\epsilon'$.
	Correspondingly, for the secrecy analysis of $\frac{1}{n}H(W_{2s}|Z^n)$, if $R_1\geq I(U_1;Z|U_2)$ then the codebook of of user 1 is equally partitioned into $2^{R_{11}}$ parts with $R_{11}=R_1-I(U_1;Z|U_2)+\epsilon'$, each part consisting of $2^{R_{12}}$ codewords with $R_{12}=I(U_1;Z|U_2)-\epsilon'$. 
	
	\subsubsection{Encoding} 
	To send message $w_{1s}$, user 1 randomly chooses $w_{1r}\in [1:2^{nR_{1r}}],$ finds $u_1^n(w_{1s}, w_{1r})$, generates $x_1^n$ according to $p(x_1|u_1)$ and sends $x_1^n$ to the channel. Similarly, to send message $w_{2s}$, user 2 randomly chooses $w_{2r}\in [1:2^{nR_{2r}}],$ finds $u_2^n(w_{2s}, w_{2r})$, generates $x_2^n$ according to $p(x_2|u_2)$ and sends $x_2^n$ to the channel.
	
	\subsubsection{Decoding}
	User 1 declares that $\hat{w}_{2s}$ is sent by user 2 if $u_2^n(\hat{w}_{2s}, \hat{w}_{2r})$ is the unique sequence such that $(u_2^n(\hat{w}_{2s}, \hat{w}_{2r}), x_1^n, y_1^n)\in T_{\epsilon}^n$.
	User 2 declares that $\hat{w}_{1s}$ is sent by user 1 if $u_1^n(\hat{w}_{1s}, \hat{w}_{1r})$ is the unique sequence such that $(u_1^n(\hat{w}_{1s}, \hat{w}_{1r}), x_2^n, y_2^n)\in T_{\epsilon}^n$.
	
	\subsubsection{Reliability Analysis}	
	Based on the AEP and packing lemma \cite{el2011network}, for sufficiently large $n$, the average error probability of $P_{e,1}$ and $P_{e,2}$ goes to zero, if
		\begin{equation} \label{eqn: ReCon R2}
			R_1\leq I(U_1;Y_2|X_2)-4\epsilon, \qquad
			R_2\leq I(U_2;Y_1|X_1)-4\epsilon.
		\end{equation}
	
	\subsubsection{individual secrecy analysis}
	Firstly, we consider the equivocation of $W_{1s}$ as follows.
		\begin{align}
			& H(W_{1s}|Z^n)\nonumber \\
			= & H(W_{1s}, W_{1r}, W_{2s}, U_1^n, U_2^n|Z^n) -H(W_{1r},W_{2s}, U_1^n, U_2^n|W_{1s},Z^n) \nonumber\\
			\stackrel{(a)}= & H(W_{1s}, W_{1r}, W_{2s}, U_1^n, U_2^n) -I(U_1^n, U_2^n;Z^n) -H(W_{1r}, W_{2s}, U_1^n, U_2^n|W_{1s}, Z^n) \nonumber\\
			\stackrel{(b)}=&  n(R_{1s}+R_{1r}+R_2) - nI(U_1, U_2;Z)-n\varepsilon_n  -H(W_{1r}, W_{2s}, U_1^n, U_2^n|W_{1s}, Z^n)\label{eqn: H(W_{1s}|Z^n)}
		\end{align}
	where $(a)$ follows from the Markov chain $(W_{1s}, W_{1r}, W_{2s})\rightarrow (U_1^n, U_2^n)\rightarrow Z^n$, such that $I(W_{1s}, W_{1r}, W_{2s}, U_1^n, U_2^n;Z^n)=I(U_1^n, U_2^n;Z^n)$;
	$(b)$ follows from $H(W_{1s}, W_{1r}, W_{2s}, U_1^n, U_2^n)=H(U_1^n, U_2^n)$, hence according to the codebook construction 
		\begin{equation*} 
			H(W_{1s}, W_{1r}, W_{2s}, U_1^n, U_2^n)=n(R_{1s}+R_{1r}+R_2).
		\end{equation*}
	And $I(U_1^n, U_2^n; Z^n)\leq nI(U_1, U_2;Z)+n\varepsilon_n$, which follows a similar proof of \cite[Lemma 3]{liu2008discrete}.
	
	Then the last term $H(W_{1r}, W_{2s}|W_{1s}, Z^n)$ in  \eqref{eqn: H(W_{1s}|Z^n)} can be bounded in two different cases as follows. 
	
		\begin{enumerate}
			\item If $R_2\leq I(U_2;Z|U_1)$,  then
			\begin{align}
				H(W_{1r}, W_{2s}, U_1^n, U_2^n|W_{1s}, Z^n)\stackrel{(c)} \leq n\epsilon', \label{Eqn: Term3Case1}
			\end{align}	
			where $(c)$ follows from the Fano's inequality by taking 	
			\begin{align}
				R_{1r}+R_2\leq I(U_1,U_2;Z)-\epsilon. \label{eqn: SeCon 11}
			\end{align}	
			Replacing the third terms in  \eqref{eqn: H(W_{1s}|Z^n)} by  \eqref{Eqn: Term3Case1}, we obtain
			\begin{align*}
				& H(W_{1s}|Z^n) \\
				\geq & n[R_{1s}+R_{1r}+R_2]-nI(U_1,U_2;Z)-n\epsilon_n-n\epsilon'\\
				\stackrel{(d)}{\geq } & nR_{1s}-n(\epsilon_n+\epsilon'),
			\end{align*}
			where $(d)$ follows by taking   
			\begin{align}
				R_{1r}+R_2\geq I(U_1,U_2;Z)-2\epsilon. \label{eqn: SeCon 112}
			\end{align}		    
			By \eqref{eqn: SeCon 11} and  \eqref{eqn: SeCon 112}, we have
			\begin{align}\label{Equ: R1rR2_Case1}
				R_{1r}+R_2= I(U_1,U_2;Z). 
			\end{align}			 
			\item If $R_2\geq I(U_2;Z|U_1)$,  then
			\begin{align}
				& H(W_{1r}, W_{2s}, U_1^n, U_2^n|W_{1s}, Z^n) \nonumber\\
				=& H(U_1^n, U_2^n|W_{1s}, Z^n) \nonumber\\
				=&  H(U_1^n|W_{1s}, Z^n)+ H(U_2^n|W_{1s}, Z^n,U_1^n) \nonumber\\
				%\stackrel{(a)} \leq& n\epsilon'+ H(C_2^n|W_{1s}, Z^n,C_1^n) \nonumber\\
				\stackrel{(e)} \leq& n\epsilon'+ n[R_2-I(U_2;Z|U_1)]+n\epsilon' \label{Eqn: Case2Term32}
			\end{align}
			where $(e)$ follows from the Fano's inequality by taking 	
			\begin{equation}\label{Eqn: Case2R1r1}
				R_{1r}\leq I(U_1;Z)-\epsilon.
			\end{equation}
			
			The second term $H(U_2^n|W_{1s}, Z^n,U_1^n)$ is bounded as follows. 
			Since $R_2\geq I(U_2;Z|U_1)$, we consider the codebook by rate splitting as explained in $1)$ codebook generation. 
			Therefore, we have
			\begin{align}
				& H(U_2^n|W_{1s}, Z^n,U_1^n) \nonumber \\
				=& H(W_{21}, W_{22}|W_{1s}, U_1^n, Z^n)\nonumber\\
				=& H(W_{21}|W_{1s}, W_{1r}, Z^n)+H(W_{22}|W_{21}, W_{1s}, U_1^n, Z^n)\nonumber\\
				\stackrel{(f)}{\leq}& nR_{21}+n\epsilon' \nonumber \\
				=& n[R_2-I(U_2;Z|U_1)]+n\epsilon', \nonumber %\label{eqn: term 32}
			\end{align}
			where $(f)$ follows from $H(W_{21}|W_{1s}, W_{1r}, Z^n)\leq H(W_{21})=nR_{21}$;
			$H(W_{22}|W_{21}, W_{1s}, U_1^n, Z^n)\leq n\epsilon'$ holds by the Fano's inequality by taking $R_{22}\leq I(U_2;Z|U_1)$, which is due to the structure of the codebook.
			
			Replacing the last term in  \eqref{eqn: H(W_{1s}|Z^n)} by  \eqref{Eqn: Case2Term32}, we obtain
			\begin{align*}
				& H(W_{1s}|Z^n) \\
				\geq & n(R_{1s}+R_{1r}+R_2)-nI(U_1,U_2;Z)- n\epsilon_n \\
				&-[n(R_2-I(U_2;Z|U_1))]+2n\epsilon')\\
				\stackrel{(g)}{\geq } & nR_{1s}-n(\epsilon_n+2\epsilon'),
			\end{align*}
			where $(g)$ follows from  
			\begin{align}
				R_{1r}\geq I(U_1;Z)-2\epsilon. \label{Eqn: Case2R1r2}
			\end{align}		    
			By \eqref{Eqn: Case2R1r1} and  \eqref{Eqn: Case2R1r2}, we have
			\begin{align}\label{Equ: R1rR2_Case2}
				R_{1r}= I(U_1;Z). 
			\end{align}			 
		\end{enumerate}

	Combining  \eqref{Equ: R1rR2_Case1} and  \eqref{Equ: R1rR2_Case2} into a more compact form, $R_{1r}$ can be rewritten as
		\begin{equation}\label{Eqn: One_R1r}
			R_{1r} = I(U_1;Z)+|I(U_2;Z|U_1)-R_2|^+.
		\end{equation}

	Following from a similar analysis of $\frac{1}{n}H(W_{2s}|Z^n)$, we have
		\begin{align}
		&R_{2r} = I(U_2;Z)+|I(U_1;Z|U_2)-R_1|^+. \label{Equ: Ind_R2r}
		\end{align}
	
	\subsubsection{secrecy rate analysis}
	Considering the reliability and the individual secrecy analysis, we obtained \eqref{Eqn: Relation R1 R1s}, \eqref{eqn: ReCon R2}, \eqref{Eqn: One_R1r},  \eqref{Equ: Ind_R2r}.
	After the Fourier-Motzkin elimination, the achievable secrecy rate region is the union of non-negative rate pair $(R_{1s}, R_{2s})$ satisfying 
 	\begin{equation*}
	\left\{		
	\begin{aligned}
	&R_{1s},\ R_{2s} \geq 0,  \\
	&R_{1s} \leq I(U_1; Y_2|X_2)-I(U_1;Z)-|I(U_2;Z|U_1)-I(U_2;Y_2|X_2)|^+, \\
	&R_{2s} \leq I(U_2; Y_1|X_1)-I(U_2;Z)-|I(U_1;Z|U_2)-I(U_1;Y_1|X_1)|^+.
	\end{aligned}
	\right\}	
	\end{equation*}

\section{Proof of Theorem \ref{Thm_Outer}} \label{Proof_Outer}
\begin{IEEEproof}
	
	First, we define the following auxiliary random variables to proceed to $R_{1s}.$
	\begin{align}\label{Def_U_V}
	U_i=X_2^{i-1} Y_2^{i-1} Z_{i+1}^n, V_{1i}=(W_{1s}, U_i), V_{2i}=(W_{2s}, U_i)
	\end{align}
	\begin{align}
	nR_{1s}=& H(W_{1s}) \leq H(W_{1s}|Z^n)+n\epsilon \nonumber\\
	=& H(W_{1s}|Z^n)-H(W_{1s}|Y_2^n, X_2^n)+H(W_{1s}|Y_2^n, X_2^n)+n\epsilon \nonumber\\
	= & H(W_{1s})-I(W_{1s};Z^n)-H(W_{1s})+I(W_{1s};Y_2^n, X_2^n)+H(W_{1s}|Y_2^n X_2^n)+n\epsilon \nonumber\\
	=&-I(W_{1s};Z^n)+I(W_{1s};Y_2^n, X_2^n)+H(W_{1s}|Y_2^n, X_2^n)+n\epsilon \nonumber\\
	\stackrel{(a)}\leq& I(W_{1s};X_2^n, Y_2^n)-I(W_{1s};Z^n)+n\epsilon+n\delta_n\nonumber\\
	=& \sum_{i=1}^n [I(W_{1s}; X_{2i}, Y_{2i}|X_2^{i-1}, Y_2^{i-1})
	-I(W_{1s};Z_i|Z_{i+1}^n)]+n\epsilon+n\delta_n\nonumber\\
	=& \sum_{i=1}^n [I(W_{1s},Z_{i+1}^n; X_{2i}, Y_{2i}|X_2^{i-1}, Y_2^{i-1})
	-I(Z_{i+1}^n; X_{2i}, Y_{2i}|X_2^{i-1}, Y_2^{i-1}, W_{1s})+n\epsilon+n\delta_n\nonumber\\
	&-I(W_{1s},X_2^{i-1}, Y_2^{i-1};Z_i|Z_{i+1}^n)
	+I(X_2^{i-1}, Y_2^{i-1};Z_i|Z_{i+1}^n, W_{1s})]+n\epsilon+n\delta_n\nonumber\\
	\stackrel{(b)}=& \sum_{i=1}^n [I(W_{1s}, Z_{i+1}^n; X_{2i},  Y_{2i}|X_2^{i-1}, Y_2^{i-1})-I(W_{1s}, X_2^{i-1}, Y_2^{i-1};Z_i|Z_{i+1}^n)]+n\epsilon+n\delta_n\nonumber\\
	=& \sum_{i=1}^n [I(Z_{i+1}^n; X_{2i}, Y_{2i}|X_2^{i-1},  Y_2^{i-1})+I(W_{1s}; X_{2i}, Y_{2i}|X_2^{i-1}, Y_2^{i-1}, Z_{i+1}^n)\nonumber\\
	&-I(X_2^{i-1}, Y_2^{i-1};Z_i|Z_{i+1}^n)\nonumber-I(W_{1s};Z_i|X_2^{i-1}, Y_2^{i-1} Z_{i+1}^n)]+n\epsilon+n\delta_n\nonumber\\
	\stackrel{(b)}=& \sum_{i=1}^n [I(W_{1s}; X_{2i}, Y_{2i}|X_2^{i-1}, Y_2^{i-1}, Z_{i+1}^n)-I(W_{1s};Z_i|X_2^{i-1}, Y_2^{i-1},  Z_{i+1}^n)]+n\epsilon+n\delta_n \nonumber\\
	\stackrel{(c)}=& \sum_{i=1}^n [I(V_{1i}; X_{2i},  Y_{2i}|U_i)-I(V_{1i};Z_i|U_i)]+n\epsilon+n\delta_n\nonumber\\
	\stackrel{(d)} %& \sum_{i=1}^n [I(V_1; X_{2},  Y_{2}|U)-I(V_1;Z|U)]+n\delta_n\nonumber\\
	=& n[I(V_1; X_{2},  Y_{2}|U)-I(V_1;Z|U)]+n\epsilon+n\delta_n   \nonumber
	\end{align}
	where
	$(a)$ follows by the Fano's inequality;
	%that $H(W_{1s}|Y_2^n, X_2^n)\leq n\bar{R}_{1s} P_r(\mathcal{E}_2)+1=n\delta_n\rightarrow 0$.
	$(b)$ follows from the Csisz\'{a}r sum identity \cite{el2013achievable};
	%lemma \ref{Lemma_CsiszarSum};
	$(c)$ follows from the definition $U_i=(X_2^{i-1}, Y_2^{i-1}, Z_{i+1}^n)$ and $V_{1i}=(W_{1s}, U_i)$ in \eqref{Def_UV}; 
	and $(d)$ follows from the standard procedure of introducing a time-sharing random variable.
 
 Similarly, we can obtain 
 	\begin{align*}
 	R_{2s}\leq& I(V_2; X_{1},  Y_{1}|U)-I(V_2;Z|U)+\epsilon+\delta_n  
    \end{align*}	
 
%	\begin{align*}
%	&I(V; X_{2} Y_{2}|U)-I(V;Z|U)\\
%	=&I(V_G,G; X_{2G} Y_{2G}|U_G,G)-I(V_G,G;Z_G|U_G,G)\\
%	=&I(V_G; X_{2G} Y_{2G}|U_G,G)+I(G; X_{2G} Y_{2G}|U_G,V_G,G)\\
%	&-I(V_G;Z_G|U_G,G)-I(G;Z_G|U_G,V_G,G)\\
%	\geq & I(V_G; X_{2G} Y_{2G}|U_G,G)-I(V_G;Z_G|U_G,G)
%	\end{align*}
%	The equivocation rate \eqref{Equ_Conv_R1ePf}
%	\begin{align}
%	\bar{R}_{1s}\leq & \sum_{i=1}^n [I(V_i; X_{2i} Y_{2i}|U_i)-I(V_i;Z_i|U_i)]+n\delta_n\nonumber\\
%	=& n[I(V_G; X_{2G} Y_{2G}|U_G,G)-I(V_G;Z_G|U_G,G)]+n\delta_n \nonumber\\
%	%=& I(V_G; X_{2G} Y_{2G}|U_G,G)-I(V_G;Z_G|U_G,G)+\delta_n\}\nonumber\\
%	\leq & n[I(V; X_{2}Y_{2}| U)-I(V;Z|U)]+n \delta_n \label{Equ_R1e_Bound}
%	\end{align}
\end{IEEEproof} 	 

\section{Proof of Theorem \ref{Thm_Outer_ind_BiGTW}	}\label{Sec: Pf_Ind_Outer_BGTW}
\begin{IEEEproof}	
	\begin{enumerate}
		\item For the two-way wiretap channels with $Y_1=Y_2=Z$, we first derive the outer bound of $R_{1s}$. 
		\begin{align}
		nR_{1s}\leq & H(W_{1s}|Z^n)+n\epsilon \nonumber\\
		\stackrel{(a)}=& H(W_{1s}|Y^n)-H((W_{1s}|X_2^nY^n)+H(W_{1s}|X_2^nY^n) +n\epsilon\nonumber\\
		=& I(W_{1s};X_2^n|Y^n)+H(W_{1s}|X_2^nY^n)+n\epsilon \nonumber\\
		\stackrel{(b)} \leq & I(W_{1s};X_2^n|Y^n)+n\epsilon_1+n\epsilon \nonumber\\
		\leq & I(W_{1s}X_1^n;X_2^n|Y^n)+n\epsilon_1+n\epsilon \nonumber\\
		=&I(X_1^n;X_2^n|Y^n)+I(W_{1s};X_2^n|Y^nX_1^n)+n\epsilon_1+n\epsilon \nonumber\\
		\stackrel{(c)}=&I(X_1^n;X_2^n|Y^n)+n\epsilon_1+n\epsilon \nonumber\\
		=&I(X_1^n;X_2^nY^n)-I(X_1^n;Y^n)+n\epsilon_1 +n\epsilon\nonumber\\
		\stackrel{(d)}= & I(X_1^n;Y^n|X_2^n)-I(X_1^n;Y^n)+n\epsilon_1+n\epsilon \nonumber\\
		=&\sum_{i=1}^{n} [I(X_1^n;Y_{i}|X_2^nY^{i-1})-I(X_1^n;Y_i|Y^{i-1})]+n\epsilon_1+n\epsilon \nonumber\\
		=&\sum_{i=1}^{n} [H(Y_{i}|X_2^nY^{i-1})-H(Y_{i}|X_1^nX_2^nY^{i-1})-H(Y_i|Y^{i-1})+H(Y_i|Y^{i-1}X_1^n)]+n\epsilon_1+n\epsilon \nonumber\\
		\stackrel{(e)}\leq &\sum_{i=1}^{n} [H(Y_{i}|X_{2i}Y^{i-1})-H(Y_{i}|X_{1i}X_{2i})-H(Y_i|Y^{i-1})+H(Y_i|X_{1i}Y^{i-1})]+n\epsilon_1+n\epsilon \nonumber\\
		=& \sum_{i=1}^{n} [I(X_{1i};Y_{i}|X_{2i}Y^{i-1})-I(X_{1i};Y_{i}|Y^{i-1})]+n\epsilon_1+n\epsilon \nonumber\\
		\stackrel{(f)}=& \sum_{i=1}^{n}[I(X_{1i};Y_{i}|X_{2i}Q_i,J=i)-I(X_{1i};Y_i|Q_i,J=i)]+n\epsilon_1+n\epsilon  \nonumber\\	
		\stackrel{(g)}	\leq& n[I(X_1;Y|X_2,Q)-I(X_1;Y|Q)]+n\epsilon_1 +n\epsilon  \nonumber 
		\end{align}	
		
		where $(a)$ follows from $Y_1=Y_2=Z=Y$; 
		$(b)$ follows from Fano's inequality by taking $R_1\leq I(X_1;Y_2X_2)$; 
		$(c)$ follows from the coding scheme; 
		$(d)$ follows from the independence of $X_{1i},\ X_{2i}$; 
		the first and the last term of $(e)$ follow that conditioning does not increase entropy, and the second term of $(e)$ follows that $Y_{i}$ is independent of everything else given $X_{1i},\ X_{2i}$ (Markov chain $Y^{i-1}\rightarrow (X_{1i},\ X_{2i})\rightarrow Y_{i}$);
		$(f)$ follows from the definition of $Q_i=Y^{i-1}$ and $J=i$;
		$(g)$ follows from that $J$ is uniformly distributed over $\{1,2,\ldots,n\}$. 
		
		Similarly, we can obtain the outer bound $R_{2s}\leq I(X_1;Y|X_2,Q)-I(X_1;Y|Q)$.
		
		\item For the second class of channels, first, we define the following auxiliary random variables to proceed to $R_{1s}$.
			\begin{align}\label{Def_UV}
			Q_i=Z^{i-1}.
			\end{align}
			
			\begin{align}
			nR_{1s}\leq & H(W_{1s}|Z^n)+n\epsilon \nonumber\\
			=& H(W_{1s}|Z^n)-H((W_{1s}|X_2^nY_2^nZ^n)+H(W_{1s}|X_2^nY_2^nZ^n)+n\epsilon \nonumber\\
			=& I(W_{1s};X_2^nY_2^n|Z^n)+H(W_{1s}|X_2^nY_2^nZ^n)+n\epsilon \nonumber\\
			\stackrel{(a)} \leq & I(W_{1s};X_2^nY_2^n|Z^n)+n\epsilon_1+n\epsilon \nonumber\\
			\leq & I(W_{1s}X_1^n;X_2^nY_2^n|Z^n)+n\epsilon_1+n\epsilon \nonumber\\
			=&I(X_1^n;X_2^nY_2^n|Z^n)+I(W_{1s};X_2^nY_2^n|Z^nX_1^n)+n\epsilon_1+n\epsilon \nonumber\\
			\stackrel{(b)}=&I(X_1^n;X_2^nY_2^n|Z^n)+n\epsilon_1+n\epsilon \nonumber\\
			=& H(X_1^n|Z^n)-H(X_1^n|X_2^nY_2^nZ^n)+n\epsilon_1+n\epsilon \nonumber\\
			\stackrel{(c)}=& H(X_1^n|Z^n)-H(X_1^n|X_2^nY_2^n)+n\epsilon_1+n\epsilon \nonumber\\
			=&I(X_1^n;X_2^nY_2^n)-I(X_1^n;Z^n)+n\epsilon_1+n\epsilon \nonumber\\
			=&I(X_1^n;Y_2^n|X_2^n)-I(X_1^n;Z^n)+n\epsilon_1+n\epsilon \label{Equ_In_Gau_R1eN} \\
			=& \sum_{i=1}^{n} [I(X_1^n;Y_{2i}|Y_2^{i-1}X_2^n)-I(X_1^n;Z_i|Z^{i-1})]+n\epsilon_1+n\epsilon \nonumber\\
			=& \sum_{i=1}^{n} [H(Y_{2i}|Y_2^{i-1}X_2^n)-H(Y_{2i}|Y_2^{i-1}X_1^nX_2^n)-H(Z_i|Z^{i-1})+H(Z_i|Z^{i-1}X_1^n)]+n\epsilon_1+n\epsilon \nonumber\\
			\stackrel{(d)}=& \sum_{i=1}^{n} [H(Y_{2i}|Y_2^{i-1}X_2^nZ^{i-1})-H(Y_{2i}|Y_2^{i-1}X_1^nX_2^nZ^{i-1})-H(Z_i|Z^{i-1})+H(Z_i|Z^{i-1}X_1^n)]+n\epsilon_1+n\epsilon \nonumber\\
			\stackrel{(e)}\leq& \sum_{i=1}^{n} [H(Y_{2i}|X_2^nZ^{i-1})-H(Y_{2i}|X_{1i}X_{2i}Z^{i-1})-H(Z_i|Z^{i-1})+H(Z_i|Z^{i-1}X_1^n)]+n\epsilon_1+n\epsilon \nonumber\\
			\stackrel{(f)}\leq& \sum_{i=1}^{n} [H(Y_{2i}|X_{2i}Z^{i-1})-H(Y_{2i}|X_{1i}X_{2i}Z^{i-1})-H(Z_i|Z^{i-1})+H(Z_i|Z^{i-1}X_{1i})] \label{Equ_In_Gau_R1eH}  \\
			=& \sum_{i=1}^{n}[I(X_{1i};Y_{2i}|X_{2i}Z^{i-1})-I(X_{1i};Z_i|Z^{i-1})]+n\epsilon_1+n\epsilon  \nonumber\\
			\stackrel{(g)}=& \sum_{i=1}^{n}[I(X_{1i};Y_{2i}|X_{2i}Q_i,J=i)-I(X_{1i};Z_i|Q_i,J=i)]+n\epsilon_1+n\epsilon  \nonumber\\	
			\stackrel{(h)}	\leq& n[I(X_1;Y_2|X_2,Q)-I(X_1;Z|Q)]+n\epsilon_1+n\epsilon  \nonumber 
			\end{align}	
			
			where $(a)$ follows from Fano's inequality by taking $R_1\leq I(X_1;X_2Y_2)$; 
			$(b)$ follows from the coding scheme; 
			$(c)$ follows from the degraded condition, i.e. $Z^n$ is degraded of $Y_2^n$;
			$(d)$ follows from the degraded condition, i.e. $Z^n$ is degraded of $Y_2^n$;
			the first term of $(e)$ follows that conditioning does not increase entropy, and the second term of $(e)$ follows that $Y_{2i}$ is independent of everything else given $X_{1i},\ X_{2i}$ (Markov chain $Z^{i-1}\rightarrow (X_{1i},\ X_{2i})\rightarrow Y_{2i}$);
			$(f)$ follows from that conditioning does not increase entropy; 
			$(g)$ follows from the definition of $Q_i=Z^{i-1}$ and $J=i$;
			$(h)$ follows from that $J$ is uniformly distributed over $\{1,2,\ldots,n\}$.
			
			Similarly, we can obtain 
			\begin{align}	
			R_{2s}\leq & H(Y_{1i}|X_{1i}Z^{i-1})-H(Y_{1i}|X_{1i}X_{2i}Z^{i-1})-H(Z_i|Z^{i-1})+H(Z_i|Z^{i-1}X_{2i}) \label{Equ_In_Gau_R2eH}\\
			\leq & I(X_2;Y_1|X_1,Q)-I(X_2;Z|Q). \nonumber
			\end{align}					
	\end{enumerate}
	
\end{IEEEproof}

\section{Proof of Theorem \ref{Thm_Capacity_DegGau} } \label{Sec: Proof_Def_Gau}
\subsection{Proof of the Achievability}
Let $U_1\sim N(0,(1-\alpha)P_1)$, $U_2\sim N(0,(1-\beta)P_2)$, $X'_1\sim N(0,\alpha P_1)$, $X'_2\sim N(0,\beta P_2)$, and $U_1,U_2,X'_1,X'_2$ are independent with each other. $X_1=U_1+X'_1$, $X_2=U_2+X'_2$. 
The achievability proof follows by calculating the mutual information terms in Theorem \ref{Thm_Inner_Individual} with the above definitions. 
Hence, the achievable secrecy rate region is 
\begin{equation} 
\mathcal{R}^{In}_{GTW}\stackrel{\vartriangle}{=}  \bigcup_{\alpha, \beta \in [0,1]} 
\left\{		
\begin{aligned}
&(R_{1s}, R_{2s}): \\
&R_{1s} \geq 0, R_{2s} \geq 0,  \\
&R_{1s}\leq \frac{1}{2}\log  \frac{(P_1+N_2)(\alpha P_1+P_2+N_e)}{(\alpha P_1+N_2)(P_1+P_2+N_e)}, \\
&R_{2s} \leq \frac{1}{2}\log \frac{(P_2+N_1)(P_1+\beta P_2+N_e)}{(\beta P_2+N_1)(P_1+P_2+N_e)}.
\end{aligned}
\right\}	
\end{equation}	

Further considering the convex hull operation, the maximum achievable secrecy rate region $\mathcal{R}^{In}_{GTW}$ is achieved when $\alpha=0$, $\beta=0$, i.e.
\begin{equation} 
\mathcal{R}^{In}_{GTW}\stackrel{\vartriangle}{=}   
\left\{		
\begin{aligned}
&(R_{1s}, R_{2s}): \\
&R_{1s} \geq 0, R_{2s} \geq 0,  \\
&R_{1s}\leq \frac{1}{2}\log \frac{(P_1+N_2)(P_2+N_e)}{N_2(P_1+P_2+N_e)}, \\
&R_{2s} \leq \frac{1}{2}\log \frac{(P_2+N_1)(P_1+N_e)}{N_1(P_1+P_2+N_e)}.
\end{aligned}
\right\}	
\end{equation}

\subsection{Proof of the Converse}
We further derive the outer bound on the secrecy rate region. From  \eqref{Equ_In_Gau_R1eH}, we have
\begin{align*}
nR_{1s}\leq & \sum_{i=1}^{n} [H(Y_{2i}|X_{2i}Z^{i-1})-H(Y_{i}|X_{1i}X_{2i}Z^{i-1})-H(Z_i|Z^{i-1})+H(Z_i|X_{1i}Z^{i-1})]+n\epsilon_1 \\
=&  \sum_{i=1}^{n} [H(Y_{2i}|X_{2i}Q_i)-H(Y_{2i}|X_{1i}X_{2i}Q_i)-H(Z_i|Q_i)+H(Z_i|X_{1i}Q_i)]-H(Z^n)+n\epsilon_1  \\
=&  \sum_{i=1}^{n} [H(Y_{2i}|X_{2i}Q_i)-H(Y_{2i}|X_{1i}X_{2i}Q_i)+H(Z_i|X_{1i}Q_i)]-H(Z^n)+n\epsilon_1 
\end{align*}

In order to obtain the outer bound on the secrecy rate region of Gaussian two-way wiretap channels, firstly we calculate the following series of entropy. 		
\begin{itemize}
	\item 	We derive the entropy of $H(Z_i|X_{1i}Q_i)$, $H(Z_i|X_{2i}Q_i)$.
	Firstly, 	 
	\begin{align*}
	& H(Z_i|X_{1i}Q_i) \nonumber \\
	\stackrel{(a)} \geq & h(Z_{i}|X_{1i}X_{2i}Q_i) \\
	=& \frac{1}{2}\log 2\pi eN_e 	
	\end{align*}	 
	where $(a)$ follows from that conditioning does not increase entropy.
	
	On the other hand,
	\begin{align*}
	& H(Z_i|X_{1i}Q_i)  \\
	\stackrel{(b)} \leq & h(Z_{i}|X_{1i})	 \\
	= & h(X_{1i}+X_{2i}+N_e|X_{1i})	 \\
	\leq & h(X_{2i}+N_e)  \\
	= & \frac{1}{2}\log 2\pi e (P_2+N_e)
	\end{align*}	 
	where $(b)$ follows from that conditioning does not increase entropy.
	
	Such that there exists some $\beta \in [0,1]$ such that
	\begin{align}
	H(Z_i|X_{1i}Q_i)=& \frac{1}{2}\log 2\pi e[N_e+\alpha(P_2+N_e-N_e)] \nonumber\\
	=& \frac{1}{2}\log 2\pi e(\beta P_2+N_e) \label{Equ: Gua_Outer_hzx1}	 	
	\end{align}	 		 	
	
	Similarly, we have some $\beta\in [0,1]$ such that
	\begin{align}
	H(Z_i|X_{2i}Q_i)
	= \frac{1}{2}\log 2\pi e(\alpha P_1+N_e) \label{Equ: Gua_Outer_hzx2}	 	
	\end{align}	 
	\item 	We derive the entropy of $H(Y_{2i}|X_{2i}Q_i)$ and $H(Y_{1i}|X_{1i}Q_i)$.\\	
	By the entropy power inequality, we obtain
	\begin{align*}
	2^{2h(Z_i|X_{1i}=x_{1i}, Q_i=q_i)} =& 2^{2h(Y_{1i}+Z'_{1i}|X_{1i}=x_{1i}, Q_i=q_i)} \\				
	\geq & 2^{2h(Y_{1i}|X_{1i}=x_{1i}, Q_i=q_i)} + 2^{2h(Z'_{1i}|X_{1i}=x_{1i}, Q_i=q_i)}  \\
	=& 2^{2h(Y_{1i}|X_{1i}=x_{1i}, Q_i=q_i)} + 2\pi e(N_e-N_1)			        	
	\end{align*}	 			
	That is 
	\begin{align*}
	h(Z_i|X_{1i}=x_{1i}, Q_i=q_i)	\geq \frac{1}{2}\log [2^{2h(Y_{1i}|X_{1i}=x_{1i}, Q_i=q_i)} + 2\pi e(N_e-N_1)] 	
	\end{align*}	 			
	Taking the expectation on both sides of the preceding equation, we have
	\begin{align*}
	h(Z_i|X_{1i}, Q_i)=&\mathbf{E} h(Z_i|X_{1i}=x_{1i}, Q_i=q_i) \\
	\geq &  \frac{1}{2}\mathbf{E} \log  [2^{2h(Y_{1i}|X_{1i}=x_{1i}, Q_i=q_i)} + 2\pi e(N_e-N_1)] \\
	\stackrel{(c)} \geq & \frac{1}{2} \log  [2^{2\mathbf{E}h(Y_{1i}|X_{1i}=x_{1i}, Q_i=q_i)} + 2\pi e(N_e-N_1)] \\
	= & \frac{1}{2} \log  [2^{2h(Y_{1i}|X_{1i}, Q_i)} + 2\pi e(N_e-N_1)] 
	\end{align*}
	where $(c)$ follows from Jensen's inequality.
	
	By \eqref{Equ: Gua_Outer_hzx1},  	       
	\begin{align}
	\frac{1}{2}\log 2\pi e(\beta P_2+N_e)= &h(Z_i|X_{1i}, Q_i) \geq \frac{n}{2} \log  [2^{2h(Y_{1i}|X_{1i}, Q_i)} + 2\pi e(N_e-N_1)] \nonumber\\
	\text{i.e.}\quad 2\pi e(\beta P_2+N_e) \geq & 2^{2h(Y_{1i}|X_{1i}, Q_i)} + 2\pi e(N_e-N_1) \nonumber\\
	h(Y_{1i}|X_{1i}, Q_i) \leq & \frac{1}{2}\log 2\pi e (\beta P_2+N_1)  \label{Equ: Gua_Outer_hy1x1}	 			 
	\end{align}	
	
	Similarly, we have
	\begin{align}
	h(Y_{2i}|X_{2i}, Q_i) \leq & \frac{1}{2}\log 2\pi e (\alpha P_1+N_2)	\label{Equ: Gua_Outer_hy2x2}	 
	\end{align}	
	
	\item We derive the entropy of $H(Y_{1i}|X_{1i}X_{2i}Q_i)$ and $H(Y_{2i}|X_{1i}X_{2i}Q_i)$.
	\begin{align}
	& H(Y_{1i}|X_{1i}X_{2i}Q_i)=H(Y_{1i}|X_{1i}X_{2i})=\frac{1}{2}\log 2\pi e N_1  \label{Equ: Gua_Outer_hy1x1x2}\\
	& H(Y_{2i}|X_{1i}X_{2i}Q_i)=H(Y_{2i}|X_{1i}X_{2i})=\frac{1}{2}\log 2\pi e N_2 \label{Equ: Gua_Outer_hy2x1x2}
	\end{align}
	\item 	We derive the entropy of $H(Z^n)$. 
	
	By \cite[Lemma 1]{Bergmans74} or \cite[Lemma 10]{Leung1978Gaussian}, let $g(P)=\frac{1}{2}\log (2\pi e P)$,
	\begin{align*}
	& H(X_1^n+X_2^n)=H(X_1^n)+H(X_2^n)=\frac{n}{2}\log 2\pi e (P_1+P_2) =ng(P_1+P_2) =n\upsilon.
	\end{align*}			
	Since  $Z^n=X_1^n+X_2^n+Z_e^n$, then
	\begin{align}
	H(Z^n)\geq & ng(N_e+g^{-1}(\upsilon)) \nonumber \\
	=& ng(N_e+g^{-1}(g(P_1+P_2))) \nonumber\\
	=& ng(N_e+P_1+P_2)\nonumber \\
	=&\frac{n}{2}\log 2\pi e (N_e+P_1+P_2) \label{Equ: Gua_Outer_hz}
	\end{align}			
\end{itemize}		

Hence, 
\begin{align*}
nR_{1s}\leq & \sum_{i=1}^{n} [H(Y_{2i}|X_{2i}Q_i)-H(Y_{2i}|X_{1i}X_{2i}Q_i)+H(Z_i|X_{1i}Q_i)]-H(Z^n)+n\epsilon_1 \\
\stackrel{(d)} \leq &  \frac{n}{2}\log 2\pi e (\alpha P_1+N_2)- \frac{n}{2}\log 2\pi e N_2+\frac{n}{2}\log 2\pi e(\beta P_2+N_e) -\frac{n}{2}\log 2\pi e (N_e+P_1+P_2) \\
=& \frac{n}{2}\log \frac{(\alpha P_1+N_2)(\beta P_2+N_e) }{ N_2 (N_e+P_1+P_2) }
\end{align*}
where $(d)$ follows by substituting \eqref{Equ: Gua_Outer_hy2x2}, \eqref{Equ: Gua_Outer_hy2x1x2}, \eqref{Equ: Gua_Outer_hzx1} and \eqref{Equ: Gua_Outer_hz}.

Similarly,
\begin{align*}
nR_{2s}\leq & \sum_{i=1}^{n} [H(Y_{1i}|X_{1i}Q_i)-H(Y_{1i}|X_{1i}X_{2i}Q_i)+H(Z_i|X_{2i}Q_i)]-H(Z^n)+n\epsilon_1 \\
\stackrel{(a)} 	\leq &  \frac{n}{2}\log 2\pi e (\beta P_2+N_1)- \frac{n}{2}\log 2\pi e N_1+\frac{n}{2}\log 2\pi e(\alpha P_1+N_e) -\frac{n}{2}\log 2\pi e (N_e+P_1+P_2) \\
=&  \frac{n}{2}\log \frac{(\beta P_2+N_1)(\alpha P_1+N_e)}{N_1(N_e+P_1+P_2)}
\end{align*}

Hence, the outer bound on secrecy capacity region is 
\begin{equation}
\mathcal{R}^{O}_{GTW}\stackrel{\vartriangle}{=}  \bigcup_{\alpha, \beta \in [0,1]} 
\left\{		
\begin{aligned}
&(R_{1s}, R_{2s}): \\
&R_{1s} \geq 0, R_{2s} \geq 0,  \\
&R_{1s}\leq \frac{1}{2}\log \frac{(\alpha P_1+N_2)(\beta P_2+N_e) }{ N_2 (N_e+P_1+P_2) }, \\
&R_{2s} \leq \frac{1}{2}\log \frac{(\beta P_2+N_1)(\alpha P_1+N_e)}{N_1(N_e+P_1+P_2)}.
\end{aligned}
\right\}	
\end{equation}	

Further considering the convex hull of $\alpha, \beta \in [0,1]$, the outer bound on the secrecy rate region is rewritten by $\alpha= \beta=1$ as
\begin{equation} 
\mathcal{R}^{O}_{GTW}\stackrel{\vartriangle}{=}   
\left\{		
\begin{aligned}
&(R_{1s}, R_{2s}): \\
&R_{1s} \geq 0, R_{2s} \geq 0,  \\
&R_{1s}\leq \frac{1}{2}\log \frac{(P_1+N_2)(P_2+N_e) }{ N_2 (N_e+P_1+P_2) }, \\
&R_{2s} \leq \frac{1}{2}\log \frac{(P_2+N_1)(P_1+N_e)}{N_1(N_e+P_1+P_2)}.
\end{aligned}
\right\}	
\end{equation}

\renewcommand\refname{reference}
\bibliographystyle{IEEEtran}
\bibliography{TwoWay}

\begin{thebibliography}{10}
\providecommand{\url}[1]{#1}
\csname url@samestyle\endcsname
\providecommand{\newblock}{\relax}
\providecommand{\bibinfo}[2]{#2}
\providecommand{\BIBentrySTDinterwordspacing}{\spaceskip=0pt\relax}
\providecommand{\BIBentryALTinterwordstretchfactor}{4}
\providecommand{\BIBentryALTinterwordspacing}{\spaceskip=\fontdimen2\font plus
\BIBentryALTinterwordstretchfactor\fontdimen3\font minus
  \fontdimen4\font\relax}
\providecommand{\BIBforeignlanguage}[2]{{%
\expandafter\ifx\csname l@#1\endcsname\relax
\typeout{** WARNING: IEEEtran.bst: No hyphenation pattern has been}%
\typeout{** loaded for the language `#1'. Using the pattern for}%
\typeout{** the default language instead.}%
\else
\language=\csname l@#1\endcsname
\fi
#2}}
\providecommand{\BIBdecl}{\relax}
\BIBdecl

\bibitem{wyner1975wire}
A.~D. Wyner, ``The wire-tap channel,'' \emph{The bell system technical
  journal}, vol.~54, no.~8, pp. 1355--1387, 1975.

\bibitem{csiszar1978broadcast}
I.~Csisz{\'a}r and J.~Korner, ``Broadcast channels with confidential
  messages,'' \emph{IEEE transactions on information theory}, vol.~24, no.~3,
  pp. 339--348, 1978.

\bibitem{Leung1978Gaussian}
S.~Leung-Yan-Cheong and M.~E. Hellman, ``The gaussian wire-tap channel,''
  \emph{Information Theory IEEE Transactions on}, vol.~24, no.~4, pp. 451--456,
  1978.

\bibitem{liang08multiple}
Y.~Liang and H.~V. Poor, ``Multiple-access channels with confidential
  messages,'' \emph{IEEE Transactions on Information Theory}, vol.~54, no.~3,
  pp. 976--1002, 2008.

\bibitem{lai2008relay}
L.~Lai and H.~El~Gamal, ``The relay--eavesdropper channel: Cooperation for
  secrecy,'' \emph{IEEE Transactions on Information Theory}, vol.~54, no.~9,
  pp. 4005--4019, 2008.

\bibitem{dai2015multiple}
B.~Dai and Z.~Ma, ``Multiple-access relay wiretap channel,'' \emph{IEEE
  Transactions on Information Forensics and Security}, vol.~10, no.~9, pp.
  1835--1849, 2015.

\bibitem{shannon1961two}
C.~E. Shannon, ``Two-way communication channels,'' in \emph{Proc. 4th Berkeley
  Symp. Math. Stat. Prob}, vol.~1.\hskip 1em plus 0.5em minus 0.4em\relax
  Citeseer, 1961, pp. 611--644.

\bibitem{tekin2007gaussian}
E.~Tekin and A.~Yener, ``The gaussian multiple access wire-tap channel:
  wireless secrecy and cooperative jamming,'' in \emph{2007 Information Theory
  and Applications Workshop}.\hskip 1em plus 0.5em minus 0.4em\relax IEEE,
  2007, pp. 404--413.

\bibitem{tekin2008general}
------, ``The general gaussian multiple-access and two-way wiretap channels:
  Achievable rates and cooperative jamming,'' \emph{IEEE Transactions on
  Information Theory}, vol.~54, no.~6, pp. 2735--2751, 2008.

\bibitem{tekin2010correction}
------, ``Correction to: the gaussian multiple access wire-tap channel and the
  general gaussian multiple access and two-way wire-tap channels: achievable
  rates and cooperative jamming,'' \emph{IEEE Transactions on Information
  Theory}, vol.~56, no.~9, pp. 4762--4762, 2010.

\bibitem{el2013achievable}
A.~El~Gamal, O.~O. Koyluoglu, M.~Youssef, and H.~El~Gamal, ``Achievable secrecy
  rate regions for the two-way wiretap channel,'' \emph{IEEE Transactions on
  Information Theory}, vol.~59, no.~12, pp. 8099--8114, 2013.

\bibitem{he2013role}
X.~He and A.~Yener, ``The role of feedback in two-way secure communications,''
  \emph{IEEE Transactions on Information Theory}, vol.~59, no.~12, pp.
  8115--8130, 2013.

\bibitem{pierrot2011strongly}
A.~J. Pierrot and M.~R. Bloch, ``Strongly secure communications over the
  two-way wiretap channel,'' \emph{IEEE Transactions on Information Forensics
  and Security}, vol.~6, no.~3, pp. 595--605, 2011.

\bibitem{bhattad2005weakly}
K.~Bhattad and K.~R. Narayanan, ``Weakly secure network coding,'' \emph{NetCod,
  Apr}, vol. 104, 2005.

\bibitem{Kadhe2014Weakly}
S.~Kadhe and A.~Sprintson, ``Weakly secure regenerating codes for distributed
  storage,'' in \emph{International Symposium on Network Coding}, 2014, pp.
  1--6.

\bibitem{Cai2002Secure}
N.~Cai and R.~W. Yeung, ``Secure network coding,'' in \emph{IEEE International
  Symposium on Information Theory, 2002. Proceedings}, 2002, p. 323.

\bibitem{ulukus2009cooperative}
S.~Ulukus and E.~Ekrem, ``Cooperative secrecy in wireless communications,'' in
  \emph{Securing Wireless Communications at the Physical Layer}.\hskip 1em plus
  0.5em minus 0.4em\relax Springer, 2009, pp. 143--172.

\bibitem{el2011network}
A.~El~Gamal and Y.-H. Kim, \emph{Network information theory}.\hskip 1em plus
  0.5em minus 0.4em\relax Cambridge university press, 2011.

\bibitem{liu2008discrete}
R.~Liu, I.~Maric, P.~Spasojevic, and R.~D. Yates, ``Discrete memoryless
  interference and broadcast channels with confidential messages: Secrecy rate
  regions,'' \emph{IEEE Transactions on Information Theory}, vol.~54, no.~6,
  pp. 2493--2507, 2008.

\bibitem{Bergmans74}
P.~P. Bergmans, ``A simple converse for broadcast channels with additive white
  gaussian noise (corresp.),'' \emph{{IEEE} Trans. Information Theory},
  vol.~20, no.~2, pp. 279--280, 1974.

\end{thebibliography}

\clearpage

\end{document}